%% file: Path-Decomposition.tex
\documentclass[a4paper]{article}

\usepackage{a4wide}
\usepackage[titlenumbered,algosection]{algorithm2e}
\usepackage{amsmath}
\usepackage{amssymb}
\usepackage{amsthm}
\usepackage{afterpage} % for \afterpage{\clearpage}
\usepackage{cancel}
\usepackage{enumitem}
\usepackage{etoolbox}
\usepackage{graphicx}
\usepackage[utf8]{inputenc}
\usepackage{subcaption}
\usepackage[normalem]{ulem}
\usepackage{verbatim}
\usepackage{textcomp}
\usepackage{xparse}
\usepackage{xcolor}
\usepackage{xspace}
\usepackage{todonotes}

\SetKwInput{AlgoInput}{\textbf{Input}}
\SetKwInput{AlgoOutput}{\textbf{Output}}
\DontPrintSemicolon

\makeatletter
\newcommand{\mylabel}[2]{#2\def\@currentlabel{#2}\label{#1}}
\makeatother

\usepackage[labelformat=simple]{subcaption}
\renewcommand\thesubfigure{(\alph{subfigure})}

\usepackage{tikz}
\usetikzlibrary{arrows}
\usetikzlibrary{decorations.pathmorphing,decorations.pathreplacing}
\usetikzlibrary{positioning}
\tikzset{snake it/.style={->,> = latex',decorate, decoration={snake, segment length=3mm, amplitude=0.5mm}}}
\renewcommand*{\thesubfigure}{(\arabic{subfigure})}

\usepackage{hyperref}

\theoremstyle{plain}
\newtheorem{theorem}{Theorem}
\newtheorem{corollary}[theorem]{Corollary}
\newtheorem{lemma}[theorem]{Lemma}
\newtheorem{definition}[theorem]{Definition}
\newtheorem{proposition}[theorem]{Proposition}

\newcommand{\keywords}[1]{\noindent\textbf{Keywords:} #1}

\DeclareDocumentCommand\R{}{\mathbb{R}}
\DeclareDocumentCommand\setdef{mo}{\left\{#1\IfNoValueTF{#2}{}{ \mid #2}\right\}}
\DeclareDocumentCommand\cplxP{}{\mathsf{P}}
\DeclareDocumentCommand\cplxNP{}{\mathsf{NP}}
\DeclareDocumentCommand\orderO{m}{\mathcal{O}\left(#1\right)}
\DeclareMathOperator{\rootGraphOperator}{RG}

\title{The Almost-Disjoint $2$-Path Decomposition Problem}
\author{Annika Thome\footnote{Lehrstuhl für Operations Research, RWTH Aachen University, Kackertstraße~7, 52072 Aachen, Germany, thome@or.rwth-aachen.de} \and Matthias Walter\footnote{Department of Applied Mathematics, University of Twente, P.O. Box 217, 7500 AE Enschede, The Netherlands}}

\date{\today}

\begin{document}

\maketitle

\DeclareDocumentCommand\probThreeSAT{}{\ensuremath{\text{3-SAT}}\xspace}

\begin{abstract}
   We consider the problem of decomposing a given (di)graph into paths of length $2$ with the additional restriction that no two such paths may have more than one vertex in common.
   We establish its $\cplxNP$-hardness by a reduction from $\probThreeSAT$, characterize (di)graph classes for which the problem can be be reduced to the Stable-set problem on claw-free graphs and describe a dynamic program for solving it for series-parallel digraphs.
\end{abstract}


\input{introduction}

\input{complexity}

\input{claw-free}

\input{series-parallel}

\pagebreak
\bibliographystyle{plain}
\bibliography{path-decomposition}

\end{document}

%% file: introduction.tex
\DeclareDocumentCommand\X{}{\mathcal{X}}
\DeclareDocumentCommand\Y{}{\mathcal{Y}}

\section{Introduction}

In this paper we investigate a restricted $2$-path decomposition problem, the problem of partitioning the edge set of a graph (or the arc set of a digraph) into paths of length $2$, or observing that no such partitioning exists.
The unrestricted problem has been shown to be polynomial-solvable by Teypaz and Rapine~\cite{teypaz:pathDecomposition}.
Moreover, they showed that even the \emph{weighted version}, for which each $2$-path has a weight and the sum of the realized weights is to be minimized, can be solved in polynomial time.

\DeclareDocumentCommand\probAlmostDisjointTwoPathDecomposition{}{AD2PD\xspace}

Our variant, the \emph{almost-disjoint $2$-path decomposition problem} (\emph{\probAlmostDisjointTwoPathDecomposition}) has the additional requirement that no pair of $2$-paths may intersect in more than one vertex.
This problem arises when organizing a so-called \emph{consecutive dinner}.
This is an event where teams are served a (typically three course) menu during the evening, where each course is served at a different location and where each team prepares one of the courses at their place.
For each course a total of three teams meet, one being the host of the current course. 
Teams that eat, e.g., the main course at the same location should not meet again to eat dessert and also should not have met for appetizers already either.
When modeling this as a decomposition problem on a graph, each vertex represents a location of a course and each arc represents the option of having consecutive courses at the respective locations.
A $2$-path then represents the route of a team where the first vertex represents the location of the appetizers, the second vertex that of the main course and the last vertex the location of the dessert.
The requirement that no pair of $2$-paths intersects in more than one vertex represents the fact that no two teams shall meet more than once during the event. 
Organizers of such an event are interested in a partition into 2-paths of a graph representing the participating teams. 
If they are further interested in ensuring that the total distance covered by all teams is minimized (a smaller distance to the next location implies the team can spend more time at the current location before they have to leave) the organizers are interested in a partition of the weighted \probAlmostDisjointTwoPathDecomposition where the weight of an edge corresponds to the distance between the corresponding locations.
Organizers of such events are for example the team of RudiRockt \cite{RudiRockt}, Running Dinner Wien \cite{RunningDinnerWien} or the AEGEE who call this event Run\&Dine.

\DeclareDocumentCommand\twoPaths{}{\ensuremath{\mathcal{P}}\xspace}
\DeclareDocumentCommand\almostDisjointTwoPathDecomposition{}{\ensuremath{\mathbf{P}^*_2\text{-decomposition}}\xspace}
\DeclareDocumentCommand\almostDisjointTwoPathDecompositions{}{\ensuremath{\mathbf{P}^*_2\text{-decompositions}}\xspace}

\paragraph{Notation and basic concepts.}
The given multi-graphs and multi-digraphs we consider will typically be denoted by $G = (V,E)$ and $D = (V,A)$, respectively, and we consider paths as subsets of edges (resp.\ arcs).
By $E(A)$ we denote the undirected version of an arc set $A$, i.e., the edge set of the underlying undirected multi-graph, denoted by $G(D) := (V, E(A))$.
Note that a pair of (anti)parallel arcs yields parallel edges.

The set $\twoPaths := \twoPaths(G)$ (resp.\ $\twoPaths(D)$) shall denote the set of $2$-paths in $G$ (resp.\ $D$).
We say that two paths $P, Q \in \twoPaths$ are \emph{in conflict} if $|V(P) \cap V(Q)| \geq 2$, where $V(P)$ shall denote the vertex set of a path $P$.
Note that, by definition, two paths with a common arc or edge are in conflict.
We call a set $\X = \setdef{ P_1, \dotsc, P_k }$ a \emph{\almostDisjointTwoPathDecomposition} if for each $i \in \setdef{1, \dotsc, k}$, we have $P_i \in \twoPaths$, the $P_i$ partition $E$, and no two paths in $\X$ are in conflict.
The \probAlmostDisjointTwoPathDecomposition problem is the problem of deciding whether a \almostDisjointTwoPathDecomposition of $G$ (resp.\ $D$) exists.
In the weighted version, we are given costs $c_P \in \R$ for each $2$-path $P \in \twoPaths$ and have to determine a \almostDisjointTwoPathDecomposition $\X = \setdef{ P_1, \dotsc, P_k }$ whose cost $\sum_{i=1}^k c(P_i)$ is minimum.

\DeclareDocumentCommand\lineGraph{}{\ensuremath{L}\xspace}
\DeclareDocumentCommand\forbiddenPairs{}{\ensuremath{\mathcal{F}}\xspace}
\DeclareDocumentCommand\conflictGraph{}{\ensuremath{\mathcal{H}}\xspace}

For a graph $G = (V,E)$ (resp.\ digraph $D = (V,A)$) we define the \emph{line graph} $\lineGraph(G) := (E,\twoPaths(G))$ (resp.\ $\lineGraph(D) := (A,\twoPaths(D))$) as the graph with vertex set $E$ (resp.\ $A$) in which two vertices $e$ and $f$ are connected by an edge if and only if $\setdef{e,f}$ is a (directed) $2$-path, i.e., $\setdef{e,f} \in \twoPaths$.
It is easy to see that $2$-path decompositions in a (di)graph correspond to perfect matchings in its line graph.
This observation can be used to solve the \probAlmostDisjointTwoPathDecomposition for graphs without small cycles.

\begin{proposition}
   \label{TheoremGirth}
   The (weighted) \probAlmostDisjointTwoPathDecomposition problem on (di)graphs of girth at least 5 reduces to the $2$-path decomposition problem.
   Consequently, it can be solved in polynomial time.
\end{proposition}

\begin{proof}
   Let $G$ (resp.\ $D$) be a (di)graph of girth at least 5, i.e., every cycle has length at least 5.
   Thus, no two $2$-paths can share more than one vertex, and the (weighted) \probAlmostDisjointTwoPathDecomposition problem is in fact a (weighted) $2$-path decomposition problem.
   The latter can be solved in polynomial time by solving a (minimum-weight) perfect matching problem in $\lineGraph(G)$~\cite{teypaz:pathDecomposition}.
\end{proof}

If shorter cycles are present, the requirement of not having conflicting paths can be interpreted as forbidden edge pairs in $L(G)$.
More precisely, using
\begin{gather}
   \label{EquationForbiddenPairs}
   \forbiddenPairs := \setdef{ \setdef{P,Q} }[ \text{$P,Q \in \twoPaths$, $P \neq Q$, and $P$ and $Q$ are in conflict} ]
\end{gather}
we can restate the perfect-matching problem in $L(G)$ with forbidden edge pairs as a maximum stable-set problem in the \emph{conflict graph} $\conflictGraph := (\twoPaths,\forbiddenPairs)$.

\paragraph{Related work.}
The more general perfect-matching problem for arbitrary graphs $\hat{G}$ with pairwise forbidden edges is known to be $\cplxNP$-hard~\cite{darmann:pathsTreesMatchingUnderDisjunctiveConstraints}.
In fact, it is shown that it is $\cplxNP$-hard even if each edge is in conflict with at most one other edge.
The same problem was studied for special graphs in~\cite{oencan:minCostPerfectMatchingProblemWithConflictPairConstraints}.
The authors considered graphs $\hat{G}$ having subgraphs $\hat{G}_1, \hat{G}_2, \dotsc, \hat{G}_k$ such that the union of arbitrary perfect matchings in each subgraph constitute a perfect matching in $\hat{G}$ and exactly one edge in each subgraph is in conflict with some other edge.
In this case, the perfect-matching problem with pairwise forbidden edges is polynomial-time equivalent to the maximum stable-set problem on the corresponding conflict graph.
The authors also showed that if the conflict graph is the union of a fixed number of disjoint cliques, then both problems can be solved in polynomial time. 

The complexity of partitioning the edges of an undirected graph into $s$ disjoint multiples of a given graph $G$ is studied in~\cite{lonc:edgeDecompositionIsPolynomial}.
Since this problem is $\cplxNP$-complete if $G$ is a complete graph with at least 3 vertices, a path with at least 5 vertices or if $G$ is a cycle of any length (see~\cite{holyer:NPCompletenessOfSomeEdgePartitionProblems}), and polynomially solvable if $G$ is a matching (see~\cite{alon:NoteOnDecompositionOfGraphsIntoIsomorphicMatchings}) or the vertex-disjoint union of a $2$-path and single edges (see~\cite{priesler:OnTheDecompositionOfGraphsIntoCopiesOfPathsAndEdges}), the authors provide a polynomial-time algorithm for the case when $G$ is the vertex-disjoint union of $2$-paths and single edges.

The problem of partitioning the edges of a graph $G$ into paths whose lengths are between two given integers $a$ and $b$ are studied in~\cite{teypaz:pathDecomposition}.
The authors propose a polynomial-time algorithm for the problem of finding such a partition if $a = 2$.
Furthermore, they investigate special cases for bipartite graphs, trees, and Euler graphs.
They also extend some of their results to the weighted case.

% The problem of finding a partition of the \emph{vertices} instead of the edges (or arcs) of a (di)graph is, to the best of our knowledge, far better investigated.
% Several $\cplxNP$-complete vertex-partitioning problems such as \todo{vertex-partitioning problems in Garey and Johnson}. \cite{karp:RecucibilityAmongCombinatorialProblems} show that the problem of vertex-partitioning into cliques is $\mathcal{NP}$-complete. \cite{kirkpatrick:complexityGeneralGraphProblems} showed that deciding whether there exists a partition of vertices such that each subgraph induced by a partition is isomorphic to a graph~$H$ is $\mathcal{NP}$-complete if~$H$ has at least 3 vertices. 

In his dissertation~\cite{fox:forbiddenPairsMakeProblemsHard}, the author discusses how forbidden pairs often result in hard problems.
In his work he focuses on combinatorial optimization problems where a forbidden pair constitutes two vertices.
For the problems of finding a matching or deciding whether a given graph is cyclic the author specifically proves that they are $\cplxNP$-complete.
He generalizes this by showing that deciding if a given graph has a certain property (a \emph{minor-ancestral} graph property) is $\cplxNP$-complete if there are certain restrictions on the graph (having no \emph{clawful characteristic graph}).

\paragraph{Contribution and outline.}
In Section~\ref{SectionComplexity} we establish the $\cplxNP$-hardness of the \probAlmostDisjointTwoPathDecomposition problem.
The remaining sections are dedicated to special (di)graph classes for which the problem can be solved in polynomial time.
In Section~\ref{SectionClawFree} we characterize those (di)graphs for which the (weighted) problem can be reduced to a stable-set problem in a claw-free graph.
Then, a dynamic program for the problem in series-parallel digraphs is proposed in Section~\ref{SectionSeriesParallel}.

%% file: complexity.tex
%%%%%%%%%%%%%%%%%%%%%%%%%%%%%%%%%%%%%%%%%%%%%
%
% COMPLEXITY
%
%%%%%%%%%%%%%%%%%%%%%%%%%%%%%%%%%%%%%%%%%%%%%

\section{Complexity}
\label{SectionComplexity}

\DeclareDocumentCommand\True{}{\texttt{true}\xspace}
\DeclareDocumentCommand\False{}{\texttt{false}\xspace}

In this section we establish the computational complexity of the \probAlmostDisjointTwoPathDecomposition problem on directed and undirected graphs.
More precisely, we show that the \probAlmostDisjointTwoPathDecomposition problem on (di)graphs is $\cplxNP$-complete by a reduction from \probThreeSAT.

\DeclareDocumentCommand\clause{}{\ensuremath{\mathcal{C}}\xspace}

An instance of the \probThreeSAT problem consists of boolean variables $x_1, \dotsc, x_n$ and a set of \emph{clauses} $\clause_1, \dotsc, \clause_m$.
Each clause $\clause_j$ is a triple of \emph{literals} $y_{j,1}, y_{j,2}, y_{j,3} \in \setdef{ x_1, x_2, \dotsc, x_n, \bar{x}_1, \bar{x}_2, \dotsc, \bar{x}_n }$.
A literal $y_{j,k}$ is \emph{positive} if it is equal to $x_i$ for some $i$ and \emph{negative} otherwise.
The \probThreeSAT problem is to determine whether there is a boolean assignment for $x$ such that for each $j \in \setdef{1,2,\dotsc,m}$, the clause $\clause_j$ is \emph{satisfied}, i.e., at least one of the three literals of $\clause_j$ is positive and the corresponding variable is true or it is negative and the corresponding variable is false.
It is well-known that the \probThreeSAT problem is $\cplxNP$-complete (see~\cite{garey:computers_intractability}).

Given a \probThreeSAT instance $\clause_1, \dotsc, \clause_m$, the reduction works by constructing a digraph $D = (V,A)$ with the property that the existence of a \almostDisjointTwoPathDecomposition of $D$, the existence of a \almostDisjointTwoPathDecomposition of $G(D)$ and the satisfiability of the \probThreeSAT instance are equivalent.
The digraph $D$ will consist of different gadget subgraphs, namely a \emph{clause gadget} for each clause that is connected to three \emph{pair gadgets} associated with pairs of literals in that clause, and one \emph{variable gadget} for each variable, which is connected to pair gadgets corresponding to clauses that contain the variable.

\paragraph{Variable gadget.}
For a variable $x_i$ that appears in $\ell$ clauses the \emph{variable gadget} is a directed cycle $C_i$ of length $12\ell$. 
This cycle is partitioned into $2\ell$ arc-disjoint paths, each of length $6$. 
We call these paths \emph{reserved subpaths}.
Two such paths are associated with every clause the variable appears in.
Let $\clause_j$ be such a clause, and let $y_{j,k}$ for $k \in \setdef{1,2,3}$ be the literal involving $x_i$, that is, $y_{j,k} \in \setdef{x_i, \bar{x}_i}$.
The two paths correspond to pairs $\setdef{k,k'}$ for both values $k' \in \setdef{1,2,3} \setminus \setdef{k}$, and are denoted by $P_{i,j,\setdef{k,k'}}$, respectively.

Suppose we number the arcs of $P_{i,j,\setdef{k,k'}}$ with $1, \dotsc, 6$ in order of traversal.
If the literal is positive, i.e., if $y_{j,k} = x_i$, then the arcs $1, \dotsc, 4$ intersect the pair gadget that corresponds to the literal pair $\setdef{k,k'}$ of clause $\clause_j$, and the arcs $5,6$ intersect no gadget.
Otherwise, i.e., if $y_{j,k} = \bar{x}_i$, the arcs $2, \dotsc, 5$ intersect the pair gadget mentioned above, and arcs $1,6$ intersect no gadget.
The cycle itself admits two \almostDisjointTwoPathDecompositions which correspond to the two values \True and \False, respectively:
The value \True is associated with the \almostDisjointTwoPathDecomposition that consists of \almostDisjointTwoPathDecompositions of the $2\ell$ reserved subpaths.
Correspondingly, the \almostDisjointTwoPathDecomposition of which some $2$-paths intersect two of the reserved subpaths is associated with the \False assignment. 
Figure~\ref{FigureVariableGadget} depicts such a variable gadget.

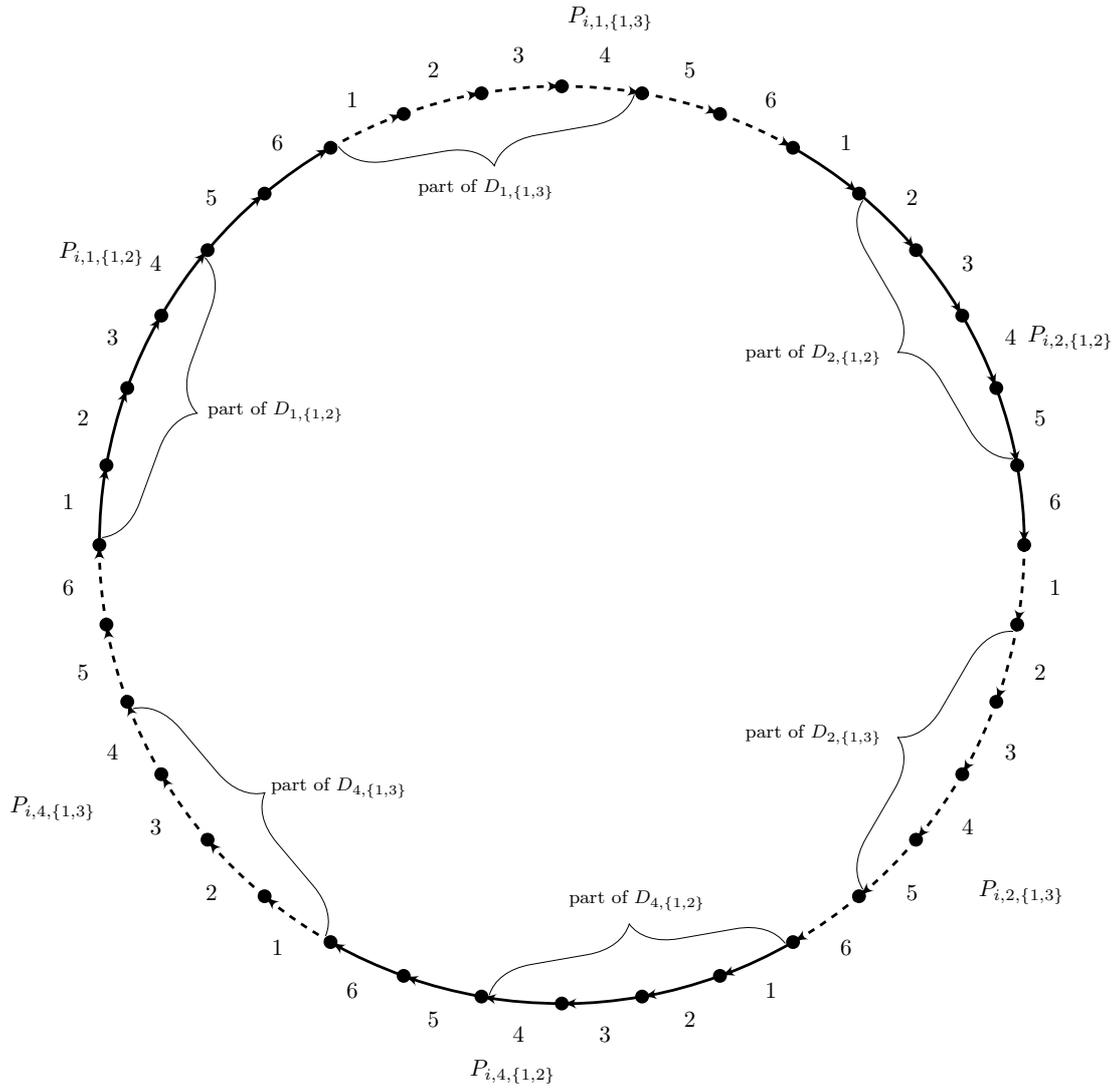
\begin{figure}
   \begin{center}
	 \resizebox{\textwidth}{!}{
      \begin{tikzpicture}
\tikzset{vertex/.style = {circle, fill, thick, inner sep=0pt, minimum size=6pt}}
\tikzset{thickvertex/.style = {draw,shape=circle,line width=1.2pt,minimum size=16pt,inner sep = 0.6pt,node distance= 1.5cm}}
\tikzset{arc/.style = {very thick, ->,> = latex'}}
\tikzset{dashedIn/.style = {densely dotted,->,> = latex'}}
\tikzset{dashedOut/.style = {densely dotted,- = latex'}}

\def \n {36}
\def \radius {7cm}
\def \radiusPathLabel {8.5cm}

\foreach \a in {1,2,3,...,\n}{
	\node[vertex] at ({-360/\n * (\a - 1)+180}:\radius) (\a) {};  %,label=below:$\a$
}
\foreach \a in {1,...,6}{
		\draw[arc] ({-360/\n * (\a - 1)+180}:\radius) 
			arc ({-360/\n * (\a - 1)+180}:{-360/\n * (\a)+180}:\radius);
		\draw[arc,dashed] ({-360/\n * (\a - 1+6)+180}:\radius) 
			arc ({-360/\n * (\a - 1+6)+180}:{-360/\n * (\a+6)+180}:\radius);
		\draw[arc] ({-360/\n * (\a - 1+12)+180}:\radius) 
			arc ({-360/\n * (\a - 1+12)+180}:{-360/\n * (\a+12)+180}:\radius);
		\draw[arc,dashed] ({-360/\n * (\a - 1+18)+180}:\radius) 
			arc ({-360/\n * (\a - 1+18)+180}:{-360/\n * (\a+18)+180}:\radius);
		\draw[arc] ({-360/\n * (\a - 1+24)+180}:\radius) 
			arc ({-360/\n * (\a - 1+24)+180}:{-360/\n * (\a+24)+180}:\radius);
		\draw[arc,dashed] ({-360/\n * (\a - 1+30)+180}:\radius) 
			arc ({-360/\n * (\a - 1+30)+180}:{-360/\n * (\a+30)+180}:\radius);
}

% arc labels
\foreach \a in {1,2,...,6}{
	\draw (\a*360/-36+185: 7.5cm) node {${\a}$};
	\draw (\a*360/-36+6*360/-36+185: 7.5cm) node {${\a}$};
	\draw (\a*360/-36+12*360/-36+185: 7.5cm) node {${\a}$};
	\draw (\a*360/-36+18*360/-36+185: 7.5cm) node {${\a}$};
	\draw (\a*360/-36+24*360/-36+185: 7.5cm) node {${\a}$};
	\draw (\a*360/-36+30*360/-36+185: 7.5cm) node {${\a}$};
}

% P_{i,1,{1,2}}
\draw (4*360/-36+185: \radiusPathLabel) node[label=below:{$P_{i,1,\setdef{1,2}}$}] () {};

% P_{i,1,{k,k'}}
\draw (10*360/-36+185: \radiusPathLabel) node[label=below:{$P_{i,1,\setdef{1,3}}$}] () {};

% P_{i,2,{1,2}}
\draw (16*360/-36+185: \radiusPathLabel) node[label=below:{$P_{i,2,\setdef{1,2}}$}] () {};

% P_{i,2,{1,3}}
\draw (22*360/-36+185: \radiusPathLabel) node[label=below:{$P_{i,2,\setdef{1,3}}$}] () {};

% P_{i,4,{1,2}}
\draw (28*360/-36+185: \radiusPathLabel) node[label=above:{$P_{i,4,\setdef{1,2}}$}] () {};

% P_{i,4,{1,3}}
\draw (34*360/-36+185: \radiusPathLabel) node[label=below:{$P_{i,4,\setdef{1,3}}$}] () {};

% identifications with clause gadgets
\draw [decorate,decoration={brace,amplitude=20pt,mirror},xshift=-20pt,yshift=10pt] 
(1) -- (5) node [black,pos=0.45,xshift=55pt] {\footnotesize part of $D_{1,\setdef{1,2}}$};

\draw [decorate,decoration={brace,amplitude=20pt,mirror},xshift=-20pt,yshift=10pt] 
(7) -- (11) node [black,pos=0.5,xshift=0pt,yshift=-30pt] {\footnotesize part of $D_{1,\setdef{1,3}}$};

\draw [decorate,decoration={brace,amplitude=20pt,mirror},xshift=-20pt,yshift=10pt] 
(14) -- (18) node [black,pos=0.6,xshift=-60pt,yshift=0pt] {\footnotesize part of $D_{2,\setdef{1,2}}$};

\draw [decorate,decoration={brace,amplitude=20pt,mirror},xshift=-20pt,yshift=10pt] 
(20) -- (24) node [black,pos=0.4,xshift=-60pt,yshift=0pt] {\footnotesize part of $D_{2,\setdef{1,3}}$};

\draw [decorate,decoration={brace,amplitude=20pt,mirror},xshift=-20pt,yshift=10pt] 
(25) -- (29) node [black,pos=0.5,xshift=0pt,yshift=30pt] {\footnotesize part of $D_{4,\setdef{1,2}}$};

\draw [decorate,decoration={brace,amplitude=20pt,mirror},xshift=-20pt,yshift=10pt] 
(31) -- (35) node [black,pos=0.65,xshift=60pt,yshift=0pt] {\footnotesize part of $D_{4,\setdef{1,3}}$};

      \end{tikzpicture}
			}
   \end{center}
   \caption{%
      Example of a variable gadget of variable $x_i$ that appears as a positive literal in clauses $C_1$ and $C_4$ and as a negative literal in clause $C_2$.
      In each of these clauses, $x_i$ is the first variable.%
   }
   \label{FigureVariableGadget}
\end{figure}

\paragraph{Pair gadget.}
Consider a clause $\clause_j$ and an unordered pair $y_{j,k},y_{j,k'}$ of its literals, i.e., $k,k' \in \setdef{1,2,3}$ and $k \neq k'$.
The pair gadget is a digraph $D_{j,\setdef{k,k'}} = (V_{j,\setdef{k,k'}}, A_{j,\setdef{k,k'}})$ with 16 vertices $v_1$, $v_2$, $v_3$, $v_4 = w_2$, $v_5$, $w_1$, $w_3$, $w_4$, $w_5$, $u_1$, $u_2$, $u_3$, $u_4$, $u_5$, $c_1$ and $c_2$, and is depicted in Figure~\ref{FigurePairGadget}.
In the figure, there are 15 dotted lines with arrows at both ends, indicating pairs of antiparallel arcs.
Additionally, six of the arcs are dashed.
We call the dashed arcs \emph{optional} arcs and denote them by $A'_{j,\setdef{k,k'}}$.
We will require the following properties of the pair gadget.

\begin{figure}
   \begin{center}
      \begin{tikzpicture}
         \tikzset{vertex/.style = {circle, draw, thick, inner sep=0pt, minimum size=6pt}}
         \tikzset{arc/.style = {->, very thick}}
         \tikzset{optional/.style = {->, very thick, red, dashed}}
         \tikzset{enforcing/.style = {<->, very thick, dotted, green!70!black}}

         \node[vertex,label=below:{$v_1$}] (v1) at (0.5,0) {};
         \node[vertex,label=left:{$v_2$}] (v2) at (0,2) {};
         \node[vertex,label=below:{$v_3$}] (v3) at (3,2) {};
         \node[vertex,label=below:{$v_4 = w_2$}] (v4w2) at (6,2) {};
         \node[vertex,label=below:{$v_5$}] (v5) at (4.9,0) {};

         \node[vertex,label=below:{$w_1$}] (w1) at (7.0,0) {};
         \node[vertex,label=below:{$w_3$}] (w3) at (9,2) {};
         \node[vertex,label=right:{$w_4$}] (w4) at (12,2) {};
         \node[vertex,label=below:{$w_5$}] (w5) at (11.5,0) {};

         \node[vertex,label=above:{$u_1$}] (u1) at (12,4) {};
         \node[vertex,label=above:{$u_2$}] (u2) at (9,4) {};
         \node[vertex,label=above:{$u_3$}] (u3) at (6,4) {};
         \node[vertex,label=above:{$u_4$}] (u4) at (3,4) {};
         \node[vertex,label=above:{$u_5$}] (u5) at (0,4) {};
         
         \node[vertex,label=left:{$c_1$}] (c1) at (8,5) {};
         \node[vertex,label=right:{$c_2$}] (c2) at (10,5) {};
         
         \draw[optional] (v1) -- (v2);
         \draw[arc] (v2) -- (v3);
         \draw[arc] (v3) -- (v4w2);
         \draw[optional] (v4w2) -- (v5);
         \draw[optional] (w1) -- (v4w2);
         \draw[arc] (v4w2) -- (w3);
         \draw[arc] (w3) -- (w4);
         \draw[optional] (w4) -- (w5);

         \draw[arc] (u1) -- (u2);
         \draw[arc] (u2) -- (u3);
         \draw[arc] (u3) -- (u4);
         \draw[arc] (u4) -- (u5);

         \draw[arc] (v2) -- (u4);
         \draw[arc] (u4) -- (v4w2);
         \draw[arc] (v4w2) -- (u2);
         \draw[arc] (u2) -- (w4);
         
         \draw[optional] (c1) -- (u2);
         \draw[optional] (u2) -- (c2);
         
         \draw[enforcing] (v1) -- (u4);
         \draw[enforcing] (v2) -- (u3);
         \draw[enforcing] (v3) -- (u4);
         \draw[enforcing] (v3) -- (u2);
         \draw[enforcing] (u4) -- (v5);
         \draw[enforcing] (w1) -- (u2);
         \draw[enforcing] (w3) -- (u2);
         \draw[enforcing] (u2) -- (w5);
         \draw[enforcing] (u3) -- (w4);
         \draw[enforcing] (u4) -- (w3);
         \draw[enforcing] (v3) to[bend right=30] (w3);
         \draw[enforcing] (v5) to[bend right=10] (w1);
         \draw[enforcing] (v3) to[bend right=15] (w1);
         \draw[enforcing] (v5) to[bend right=5] (u2);
         \draw[enforcing] (u4) to[bend right=5] (w1);
         \draw[enforcing] (v5) to[bend right=15] (w3);
      \end{tikzpicture}
   \end{center}
   \caption{Pair gadget. Dotted double arcs indicate pairs of antiparallel arcs. Only dashed arcs are part of $2$-paths that might use arcs from outside the gadget.}
   \label{FigurePairGadget}
\end{figure}
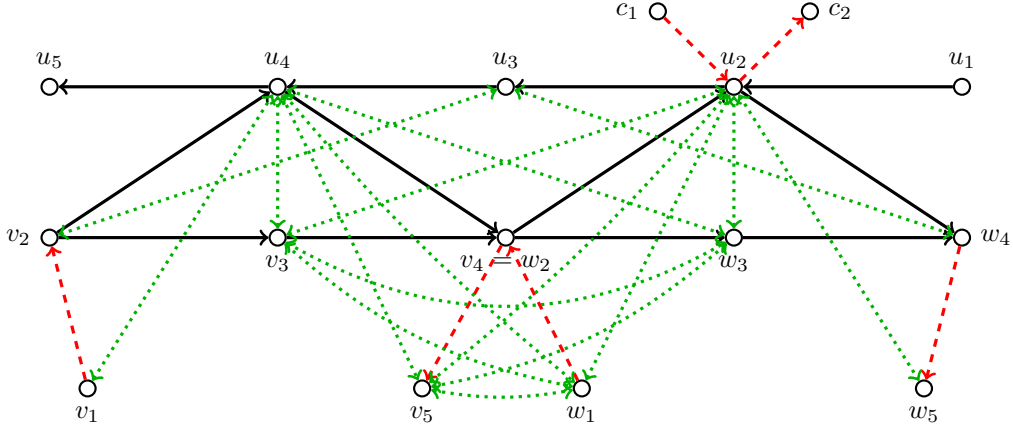

\begin{proposition}
   \label{TheoremPairGadgetDecompositions}
   There exist \almostDisjointTwoPathDecompositions of $D_{j,\setdef{k,k'}} - M$ for the following sets $M$ of optional arcs.
   \begin{itemize}
   \item[(TT)]
      \label{TheoremPairGadgetDecompositionsTT}
      $M = \setdef{(c_1,u_2), (u_2,c_2)}$
   \item[(TF)]
      \label{TheoremPairGadgetDecompositionsTF}
      $M = \setdef{(c_1,u_2), (u_2,c_2), (w_1,w_2), (w_4,w_5)}$
   \item[(FT)]
      \label{TheoremPairGadgetDecompositionsFT}
      $M = \setdef{(c_1,u_2), (u_2,c_2), (v_1,v_2), (v_4,v_5)}$
   \item[(FF)]
      \label{TheoremPairGadgetDecompositionsFF}
      $M = \setdef{(v_1,v_2), (v_4,v_5), (w_1,w_2), (w_4,w_5)}$
   \end{itemize}
\end{proposition}

\begin{proof}
   In each of the four cases, the pairs of antiparallel arcs are pairwised joined to a $2$-path.
   Hence, in a \almostDisjointTwoPathDecomposition, the two vertices of each such path may not be connected by any other $2$-path of the decomposition.
   We list the additional $2$-paths for each case as triples $(x,y,z)$ to denote the path $\setdef{(x,xy), (y,z)}$.
   \begin{itemize}
   \item[(TT)]
      $(v_1,v_2,v_3)$, $(v_3,v_4,v_5)$, $(w_1,w_2,w_3)$, $(w_3,w_4,w_5)$, $(v_2,u_4,v_4)$, $(w_2,u_2,w_4)$, $(u_1,u_2,u_3)$ and $(u_3,u_4,u_5)$
   \item[(TF)]
      $(v_1,v_2,v_3)$, $(v_3,v_4,v_5)$, $(w_2,w_3,w_4)$, $(u_1,u_2,w_4)$, $(w_2,u_2,u_3)$, $(v_2,u_4,v_4)$ and $(u_3,u_4,u_5)$
   \item[(FT)]
      $(v_2,v_3,v_4)$, $(w_1,w_2,w_3)$, $(w_3,w_4,w_5)$, $(w_2,u_2,w_4)$, $(u_1,u_2,u_3)$, $(u_3,u_4,v_4)$ and $(v_2,u_4,u_5)$
   \item[(FF)]
      $(v_2,v_3,v_4)$, $(w_2,w_3,w_4)$, $(v_2,u_4,u_5)$, $(u_3,u_4,v_4)$, $(w_2,u_2,c_2)$, $(c_1,u_2,w_4)$ and $(u_1,u_2,u_3)$
   \end{itemize}
   One now carefully verifies that these sets of $2$-paths are indeed \almostDisjointTwoPathDecompositions of $D_{j,\setdef{k,k'}} - M$ for the respective set $M$.
\end{proof}

Note that there also exist corresponding \almostDisjointTwoPathDecompositions for $G(D_{j,\setdef{k,k'}} - M)$ since each \almostDisjointTwoPathDecomposition of a digraph induces one of its associated undirected graph.

\begin{proposition}
   \label{TheoremPairGadgetImplications}
   Let $\X$ be a \almostDisjointTwoPathDecomposition of $G(D_{j,\setdef{k,k'}} - M)$ for some $M \subseteq A'_{j,\setdef{k,k'}}$.
   Then
   \begin{enumerate}[label={(\roman*)}]
   \item
      \label{TheoremPairGadgetImplicationsFirstVariable}
      Either $(v_1,v_2,v_3),(v_3,v_4,v_5) \in \X$ holds or $(v_2,v_3,v_4) \in \X$ and $(v_1,v_2),(v_4,v_5) \in M$ hold.
   \item
      \label{TheoremPairGadgetImplicationsSecondVariable}
      Either $(w_1,w_2,w_3),(w_3,w_4,w_5) \in \X$ holds or $(w_2,w_3,w_4) \in \X$ and $(w_1,w_2),(w_4,w_5) \in M$ hold.
   \item
      \label{TheoremPairGadgetImplicationsClause}
      If $(v_1,v_2),(v_4,v_5),(w_1,w_2),(w_4,w_5) \in M$ then $(c_1,u_2), (u_2,c_2) \notin M$.
   \end{enumerate}
\end{proposition}

\begin{proof}
   We denote by $(x,y,z) \in \X$ the fact that the $2$-path $\setdef{\setdef{x,y}, \setdef{y,z}}$ is contained in $\X$.
   First observe that $\X$ has to contain the $2$-paths corresponding to antiparallel arcs of $D_{j,\setdef{k,k'}}$ since the respective two parallel edges share $2$ vertices.
   For the remainder of the proof we say that vertices $x$ and $y$ are \emph{in conflict} if they are (by previous arguments) part of a $2$-path.

   To see~\ref{TheoremPairGadgetImplicationsFirstVariable}, observe that each of the vertices of the path $P := \setdef{\setdef{v_1, v_2}, \setdef{v_2,v_3}, \setdef{v_3,v_4},\setdef{v_4,v_5}}$ is in conflict with all vertices that do not belong to $P$ but that can be reached by one edge of $P$ and one edge not in $P$ that does not belong to a pair of antiparallel arcs.
   Hence, edges of the path can only be contained in $2$-paths of $\X$ if the $2$-path is a subpath of $P$.
   Since the two middle edges of $P$ are not optional, they must be paired with each other or with the two remaining edges of $P$, which proves~\ref{TheoremPairGadgetImplicationsFirstVariable}.
   The proof of~\ref{TheoremPairGadgetImplicationsSecondVariable} is similar and thus omitted.

   To see~\ref{TheoremPairGadgetImplicationsClause}, assume that there exists a \almostDisjointTwoPathDecomposition $\X$ of $G(D_{j,\setdef{k,k'}} - M)$ for some $M \subseteq A'_{j,\setdef{k,k'}}$ with $|M| > 4$.
   Since $|A_{j,\setdef{k,k'}}|$ is even, we must have $|M| = 6$, i.e., $M = A'_{j,\setdef{k,k'}}$.
   By previous arguments we have $(v_2,v_3,v_4), (w_2,w_3,w_4) \in \X$.
   Now, $v_2$ and $v_4$ are in conflict, and since $v_2$ and $u_3$ are also in conflict, we must have $(v_2,u_4,u_5) \in \X$.
   Similarly, since $w_4$ is in conflict with $w_2$ and with $u_3$, we have $(u_1,u_2,w_4) \in \X$.
   The remaining edges form a cycle of length $4$, which does not have a \almostDisjointTwoPathDecomposition,
   which contradicts the existence of $\X$ and concludes the proof.
\end{proof}

Note that $M \subseteq A'_{j,\setdef{k,k'}}$ has to satify the properties of the above proposition in the directed case as well.

\paragraph{Clause gadget.}
The clause gadget of a clause $C_j$ is a digraph $D_j = (V_j, A_j)$ with 9 vertices $c_1$, $c_2$, $u_1$, $u_2$, $v_1$, $v_2$, $p_1$, $p_2$ and $p_3$, and is depicted in Figure~\ref{FigureClauseGadget}.

\begin{figure}
   \begin{center}
      \begin{tikzpicture}
         \tikzset{vertex/.style = {circle, draw, thick, inner sep=0pt, minimum size=6pt}}
         \tikzset{arc/.style = {->, very thick}}
         \tikzset{optional/.style = {->, very thick, red, dashed}}
         \tikzset{enforcing/.style = {<->, very thick, dotted, green!70!black}}

         \node[vertex,label=left:{$c_1$}] (c1) at (0,0) {};
         \node[vertex,label=above:{$u_1$}] (u1) at (-1,2) {};
         \node[vertex,label=above:{$u_2$}] (u2) at (1,2) {};
         \node[vertex,label=right:{$c_2$}] (c2) at (3,0) {};
         \node[vertex,label=above:{$v_1$}] (v1) at (2,2) {};
         \node[vertex,label=above:{$v_2$}] (v2) at (4,2) {};
         \node[vertex,label=left:{$p_1$}] (p1) at (-0.5,-2) {};
         \node[vertex,label=below:{$p_2$}] (p2) at (1.5,-2) {};
         \node[vertex,label=right:{$p_3$}] (p3) at (3.5,-2) {};
         
         \draw[arc] (u1) -- (c1);
         \draw[arc] (u2) -- (c1);
         \draw[optional] (c1) -- (p1);
         \draw[optional] (c1) -- (p2);
         \draw[optional] (c1) -- (p3);
         \draw[arc] (c2) -- (v1);
         \draw[arc] (c2) -- (v2);
         \draw[optional] (p1) -- (c2);
         \draw[optional] (p2) -- (c2);
         \draw[optional] (p3) -- (c2);

         \draw[enforcing] (u1) -- (u2);
         \draw[enforcing] (v1) -- (v2);
         \draw[enforcing] (p1) -- (p2);
         \draw[enforcing] (p2) -- (p3);
         \draw[enforcing] (p1) to[bend right=30] (p3);
      \end{tikzpicture}
   \end{center}
   \caption{Clause gadget. Dotted double arcs indicate pairs of antiparallel arcs. Only dashed arcs are part of $2$-paths that might use arcs from outside the gadget.}
   \label{FigureClauseGadget}
\end{figure}
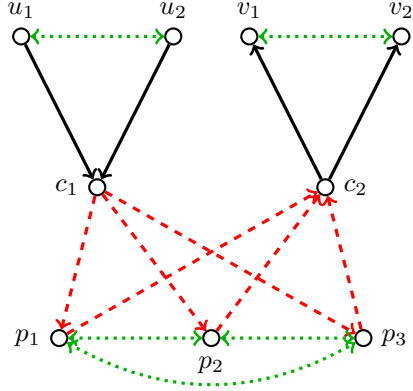

The relevant properties are captured in the following results.
\begin{proposition}
   \label{TheoremClauseGadgetDecompositions}
   There exist \almostDisjointTwoPathDecompositions of $D_j$ and of $D_j - \setdef{ (c_1,p_i), (p_i, c_2) }$ for $i = 1,2,3$, but no \almostDisjointTwoPathDecomposition of $D_j - \setdef{ (c_1,p_1), (p_1,c_2), (c_1,p_2), (p_2,c_2), (c_1,p_3), (p_3,c_2) }$.
\end{proposition}

\begin{proof}
   We again denote by $(x,y,z)$ the $2$-path $\setdef{\setdef{x,y}, \setdef{y,z}}$. 
   A \almostDisjointTwoPathDecomposition of $D_j$ is given by $(u_1,c_1,p_1)$, $(u_2,c_1,p_2)$, $(p_1,c_2,v_1)$, $(p_2,c_2,v_2)$, $(c_1,p_3,c_2)$ and the 2-paths consisting of pairs of antiparallel edges.
   This \almostDisjointTwoPathDecomposition is also a \almostDisjointTwoPathDecomposition of $D_j - \setdef{ (c_1,p_i), (p_i, c_2) }$ for $i=1,2$. 
   A \almostDisjointTwoPathDecomposition of $D_j - \setdef{ (c_1,p_i), (p_i, c_2) }$ with $i=3$ is given by $(u_1,c_1,p_3)$, $(u_2,c_1,p_2)$, $(p_3,c_2,v_1)$, $(p_2,c_2,v_2)$, $(c_1,p_1,c_2)$ and the 2-paths consisting of pairs of antiparallel edges.
   It is easy to see that there is no \almostDisjointTwoPathDecomposition of $D_j - \{ (c_1,p_1)$, $(p_1,c_2)$, $ (c_1,p_2)$, $ (p_2,c_2)$, $ (c_1,p_3)$, $ (p_3,c_2) \}$.
\end{proof}

\paragraph{Reduction.}
For a given \probThreeSAT instance with $n$ variables $x_1, \dotsc, x_n$ and $m$ clauses $\clause_1, \dotsc, \clause_m$ we construct the following digraph $D$.
For each clause $\clause_j = (y_{j,1}, y_{j,2}, y_{j,3})$ we create the clause gadget $D_j$, for each pair $(k,k') \in \setdef{ (1,2), (1,3), (2,3) }$ of literals of this clause we create the pair gadget $D_{j,\setdef{k,k'}}$, and for each variable $x_i$ we create the variable gadget $C_i$.
Now, certain vertices of these gadget graphs are identified with each other.
If two identified vertices are connected by an arc in both gadgets, then also these arcs are identified with each other.

For each pair gadget $D_{j,\setdef{k,k'}}$, the corresponding vertices $c_1$, $c_2$ and $u_2$ are identified with the vertices $c_1$, $c_2$ and $p_i$, respectively, of the clause gadget $D_j$.
Since $\clause_j$ is related to exactly three literal pairs, we can do this such that $p_1$, $p_2$ and $p_3$ of $D_j$ are each identified with the $u_2$ vertex of \emph{exactly one} pair gadget.

For each variable $x_i$ that appears in some literal $y_{j,k}$ of $\clause_j$ and for each $k' \in \setdef{1,2,3} \setminus \setdef{k}$, we identify certain vertices of $D_{j,\setdef{k,k'}}$ with certain vertices of the subpath $P_{i,j,\setdef{k,k'}}$ of $C_i$.
Denote by $q_1, q_2, \dotsc, q_7$ the nodes of $P_{i,j,\setdef{k,k'}}$ in order of traversal.
The identification depends on $k$ and $k'$ and on the type of the literal $y_{j,k}$:
\begin{itemize}
\item[]
   If $k < k'$ and $y_{j,k} = x_i$, identify vertex $v_{\ell}$ of $D_{j,\setdef{k,k'}}$ with vertex $q_\ell$ for $\ell=1,\dotsc,5$.
\item[]
   If $k < k'$ and $y_{j,k} = \bar{x}_i$, identify vertex $v_{\ell}$ of $D_{j,\setdef{k,k'}}$ with vertex $q_{\ell + 1}$ for $\ell=1,\dotsc,5$.
\item[]
   If $k > k'$ and $y_{j,k} = x_i$, identify vertex $w_{\ell}$ of $D_{j,\setdef{k,k'}}$ with vertex $q_\ell$ for $\ell=1,\dotsc,5$.
\item[]
   If $k > k'$ and $y_{j,k} = \bar{x}_i$, identify vertex $w_{\ell}$ of $D_{j,\setdef{k,k'}}$ with vertex $q_{\ell + 1}$ for $\ell=1,\dotsc,5$.
\end{itemize}

Figure~\ref{FigureEntireGadget} is a schematic representation of the entire gadget.

\begin{figure}
   \begin{center}
   \resizebox{\textwidth}{!}{
      \begin{tikzpicture}
         \tikzset{vertex/.style = {circle, draw, thick, inner sep=0pt, minimum size=6pt}}
         \tikzset{arc/.style = {->, very thick}}
         \tikzset{optional/.style = {->, very thick}}
         \tikzset{variableCycle/.style = {->, thick,dashed}}

         % pair gadgets - square borders
         \draw[thick,rounded corners=8pt] (-0.7,0) -- (-0.7,4.5) -- (4.7,4.5) -- (4.7,-1) -- (-0.7,-1) -- (-0.7,0);
         \draw[thick,rounded corners=8pt] (6+-0.7,0) -- (6+-0.7,4.5) -- (6+4.7,4.5) -- (6+4.7,-1) -- (6+-0.7,-1) -- (6+-0.7,0);
         \draw[thick,rounded corners=8pt] (12+-0.7,0) -- (12+-0.7,4.5) -- (12+4.7,4.5) -- (12+4.7,-1) -- (12+-0.7,-1) -- (12+-0.7,0);
         \node[label=below:{$D_{j,\setdef{1,2}}$}] (pairGadgetLabel1) at (0.2,3.3) {};
         \node[label=below:{$D_{j,\setdef{1,3}}$}] (pairGadgetLabel2) at (6+0.2,3.3) {};
         \node[label=below:{$D_{j,\setdef{2,2}}$}] (pairGadgetLabel3) at (12+0.2,3.3) {};

         % clause gadget - square borders
         \draw[thick,rounded corners=8pt] (2.3,4) -- (2.3,10.5) -- (15.7,10.5) -- (15.7,3.4) -- (2.3,3.4) -- (2.3,4);    
         \node[label=below:{$D_{j}$}] (clauseGadget) at (5,9.5) {};

         % first pair gadget
         \node[vertex,label=left:{$v_1$}] (v1) at (0.5,0) {};
         \node[vertex,label=left:{$v_2$}] (v2) at (0,2) {};
         \node[vertex,label=above:{$v_3$}] (v3) at (1,2) {};
         \node[vertex,label=above:{$v_4 = w_2$}] (v4w2) at (2,2) {};
         \node[vertex,label=left:{$v_5$}] (v5) at (1.5,0) {};

         \node[vertex,label=left:{$w_1$}] (w1) at (2.5,0) {};
         \node[vertex,label=above:{$w_3$}] (w3) at (3,2) {};
         \node[vertex,label=right:{$w_4$}] (w4) at (4,2) {};
         \node[vertex,label=right:{$w_5$}] (w5) at (3.5,0) {};

         \node[vertex,label=below:{$u_2=p_1$}] (u2) at (3,4) {};

         \node[vertex,label=left:{$c_1$}] (c1) at (8,8) {};
         \node[vertex,label=right:{$c_2$}] (c2) at (10,8) {};

         \draw[optional] (v1) -- (v2);
         \draw[arc] (v2) -- (v3);
         \draw[arc] (v3) -- (v4w2);
         \draw[optional] (v4w2) -- (v5);
         \draw[optional] (w1) -- (v4w2);
         \draw[arc] (v4w2) -- (w3);
         \draw[arc] (w3) -- (w4);
         \draw[optional] (w4) -- (w5);

         \draw[optional] (c1) -- (u2);
         \draw[optional] (u2) -- (c2);

         % second pair gadget        
         \node[vertex,label=left:{$v_1$}] (v12) at (6+0.5,0) {};
         \node[vertex,label=left:{$v_2$}] (v22) at (6+0,2) {};
         \node[vertex,label=above:{$v_3$}] (v32) at (6+1,2) {};
         \node[vertex,label=above:{$v_4 = w_2$}] (v4w22) at (6+2,2) {};
         \node[vertex,label=left:{$v_5$}] (v52) at (6+1.5,0) {};

         \node[vertex,label=left:{$w_1$}] (w12) at (6+2.5,0) {};
         \node[vertex,label=above:{$w_3$}] (w32) at (6+3,2) {};
         \node[vertex,label=right:{$w_4$}] (w42) at (6+4,2) {};
         \node[vertex,label=right:{$w_5$}] (w52) at (6+3.5,0) {};

         \node[vertex,label=below:{$u_2=p_2$}] (u22) at (6+3,4) {};

         \draw[optional] (v12) -- (v22);
         \draw[arc] (v22) -- (v32);
         \draw[arc] (v32) -- (v4w22);
         \draw[optional] (v4w22) -- (v52);
         \draw[optional] (w12) -- (v4w22);
         \draw[arc] (v4w22) -- (w32);
         \draw[arc] (w32) -- (w42);
         \draw[optional] (w42) -- (w52);

         \draw[optional] (c1) -- (u22);
         \draw[optional] (u22) -- (c2);

         % third pair gadget
         \node[vertex,label=left:{$v_1$}] (v13) at (12+0.5,0) {};
         \node[vertex,label=left:{$v_2$}] (v23) at (12+0,2) {};
         \node[vertex,label=above:{$v_3$}] (v33) at (12+1,2) {};
         \node[vertex,label=above:{$v_4 = w_2$}] (v4w23) at (12+2,2) {};
         \node[vertex,label=left:{$v_5$}] (v53) at (12+1.5,0) {};

         \node[vertex,label=left:{$w_1$}] (w13) at (12+2.5,0) {};
         \node[vertex,label=above:{$w_3$}] (w33) at (12+3,2) {};
         \node[vertex,label=right:{$w_4$}] (w43) at (12+4,2) {};
         \node[vertex,label=right:{$w_5$}] (w53) at (12+3.5,0) {};

         \node[vertex,label=below:{$u_2=p_3$}] (u23) at (12+3,4) {};

         \draw[optional] (v13) -- (v23);
         \draw[arc] (v23) -- (v33);
         \draw[arc] (v33) -- (v4w23);
         \draw[optional] (v4w23) -- (v53);
         \draw[optional] (w13) -- (v4w23);
         \draw[arc] (v4w23) -- (w33);
         \draw[arc] (w33) -- (w43);
         \draw[optional] (w43) -- (w53);

         \draw[optional] (c1) -- (u23);
         \draw[optional] (u23) -- (c2);

         % variable gadgets
         \draw[variableCycle] (v52) to [bend left=90] node[pos=0.75,above] {$C_1$} (v1);
         \draw[variableCycle] (v5) to [bend right=90] (v12);

         \draw[variableCycle] (v53) to [bend left=90] node[above] {$C_4$} (w1);
         \draw[variableCycle] (w5) to [bend right=90] (v13);

         \draw[variableCycle] (w53) to [bend left=90] node[pos=0.25,above] {$C_5$} (w12);
         \draw[variableCycle] (w52) to [bend right=90] (w13);
      \end{tikzpicture}
      }
   \end{center}
   \caption{%
      Schematic representation of the entire gadget for $C_j = (x_1, x_4, x_5)$.
      The upper part depicts the clause gadget, the middle part depicts the three pair gadgets, and the dashed arcs illustrate the variable cycles this gadget is connected to.
   }
   \label{FigureEntireGadget}
\end{figure}
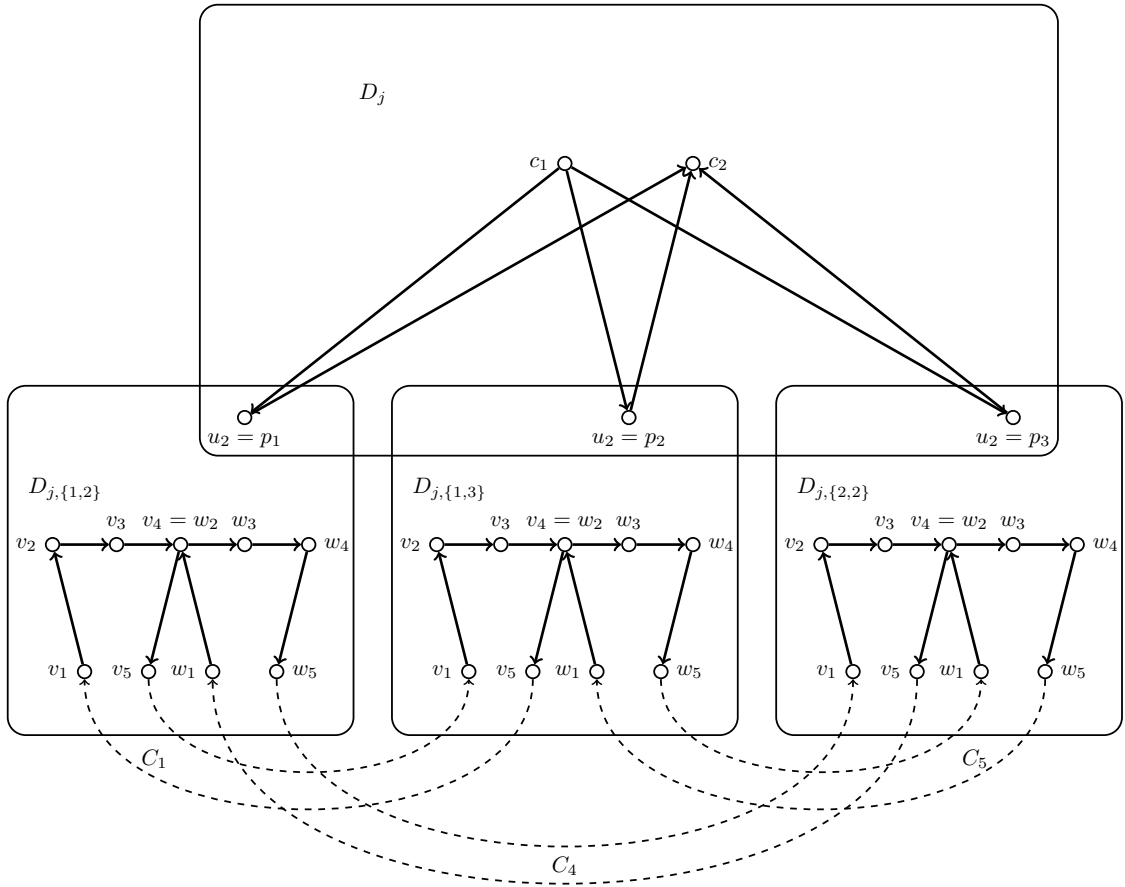

We now establish the correctness of the reduction.

\begin{lemma}
   \label{TheoremHardnessEquivalence}
   Let $D$ be the constructed digraph for a given \probThreeSAT instance.
   Then the following are equivalent.
   \begin{enumerate}[label={(\roman*)}]
   \item
      The \probThreeSAT instance is satisfiable.
   \item
      $D$ admits a \almostDisjointTwoPathDecomposition.
   \item
      $G(D)$ admits a \almostDisjointTwoPathDecomposition.
   \end{enumerate}
\end{lemma}

\begin{proof}
   Assume the \probThreeSAT is satisfiable.
   We now construct a \almostDisjointTwoPathDecomposition $\X$ of $D$.
   To this end, for each variable $x_i$ that is \True, decompose the cycle $C_i$ into $\tfrac{1}{2}|C_i|$ $2$-paths such that each of the subpaths $P_{i,j,\setdef{k,k'}}$ contains three such $2$-paths.
   For each variable $x_i$ that is \False, decompose $C_i$ into $2$-paths such that each $P_{i,j,\setdef{k,k'}}$ contains two $2$-paths and intersects exactly one arc of two other $2$-paths.

   For each variable gadget $D_{j,\setdef{k,k'}}$ (assuming $k < k'$ without loss of generality), consider the partial decompositions of Proposition~\ref{TheoremPairGadgetDecompositions} that correspond to the respective truth assignment:
   \begin{itemize}
   \item
      If $y_{j,k}$ and $y_{j,k'}$ are both \True, consider the decomposition from (TT).
   \item
      If $y_{j,k}$ is \True and $y_{j,k'}$ is \False, consider the decomposition from (TF).
   \item
      If $y_{j,k}$ is \False $y_{j,k'}$ is \True, consider the decomposition from (FT).
   \item
      If $y_{j,k}$ and $y_{j,k'}$ are both \False, consider the decomposition from (FF).
   \end{itemize}
   Due to the way the vertices of $D_{j,\setdef{k,k'}}$ and $C_i$ were identified, the decomposition of $D_{j,\setdef{k,k'}}$ aligns with that of $C_i$ defined above.

   Consider a clause $\clause_j$ with at most one \False literal.
   Neither of the considered decompositions as described in Proposition~\ref{TheoremPairGadgetDecompositions} cover the arcs $(c_1,u_2)$ and $(u_2,c_2)$ of the corresponding pair gadget.
   In this case, $\X$ shall contain the \almostDisjointTwoPathDecomposition of $D_j$, which exists by Proposition~\ref{TheoremClauseGadgetDecompositions}. 
   Note that in Proposition~\ref{TheoremClauseGadgetDecompositions} vertex $u_2$ is denoted by $p_i$ for some $i$.

   Consider any other clause $\clause_j$, which is satisfied, and thus contains exactly two \False literals $y_{j,k}$ and $y_{j,k'}$.
   Let $k'' \in \setdef{1,2,3} \setminus \setdef{k,k'}$ be the unique third literal.
   The considered decomposition of $D_{j,\setdef{k,k''}}$ and of $D_{j,\setdef{k',k''}}$ in Proposition~\ref{TheoremPairGadgetDecompositions} do not cover the respective arcs $(c_1,u_2)$ and $(u_2,c_2)$, while the considered decomposition of $D_{j,\setdef{k,k'}}$ covers both arcs.
   Let $\ell \in \setdef{1,2,3}$ be such that vertex $p_{\ell}$ of $D_j$ was identified with vertex $u_2$ of $D_{j,\setdef{k,k'}}$.
   In this case, $\X$ shall contain the \almostDisjointTwoPathDecomposition of $D_j - \setdef{(c_1,p_{\ell}), (p_{\ell}, c_2)}$, which exists by Proposition~\ref{TheoremClauseGadgetDecompositions}.
   By construction, $\X$ is indeed a \almostDisjointTwoPathDecomposition of $D$.

\bigskip
   
   It is easy to see that if $D$ admits a \almostDisjointTwoPathDecomposition then $G(D)$ admits a \almostDisjointTwoPathDecomposition by replacing each directed $2$-path by its undirected counterpart.

\bigskip

   Now assume that $G(D)$ admits a \almostDisjointTwoPathDecomposition $\X$.
   We prove that the \probThreeSAT instance is satisfiable.
   Consider a pair gadget $D_{j,\setdef{k,k'}}$.
   By construction, only the optional arcs can be part of $2$-paths of $\X$ that use arcs not belonging to $D_{j,\setdef{k,k'}}$.
   Hence, $\X$ induces a \almostDisjointTwoPathDecomposition of $D_{j,\setdef{k,k'}} - M$ for some set $M \subseteq A'_{j,\setdef{k,k'}}$ as defined before Proposition~\ref{TheoremPairGadgetDecompositions}.
   Then Proposition~\ref{TheoremPairGadgetImplications} implies that the four arcs that $D_{j,\setdef{k,k'}}$ shares with a variable gadget $C_i$ for some variable $i$ are covered by two or three consecutive $2$-paths of $C_i$.
   This implies that for each variable $x_i$, $\X$ induces a \almostDisjointTwoPathDecomposition of $C_i$.
   The \probThreeSAT solution is constructed by setting $x_i$ to \True if and only if $\X$ induces a \almostDisjointTwoPathDecomposition of each of the reserved subpaths of $C_i$.
	
   It remains to prove that the \probThreeSAT solution is feasible.
   To this end, assume that there exists a clause $\clause_j$ that is not satisfied, i.e., $y_{j,1}$, $y_{j,2}$ and $y_{j,3}$ are \False.
   For each pair $(k,k') \in \setdef{(1,2), (1,3), (2,3)}$, consider the pair gadget $D_{j,\setdef{k,k'}}$ and a set $M \subseteq A'_{j,\setdef{k,k'}}$ such that $\X \setminus M$ induces a \almostDisjointTwoPathDecomposition of $D_{j,\setdef{k,k'}} - M$.
   Due to the way the vertices of $D_{j,\setdef{k,k'}}$ were identified with the variable gadgets, we have $(v_1,v_2), (v_4,v_5), (w_1,w_2), (w_4,w_5) \in M$.
   By Proposition~\ref{TheoremPairGadgetImplications}, this implies $(c_1,u_2),(u_2,c_2) \notin M$, i.e., the two corresponding edges belong to $2$-paths of $D_{j,\setdef{k,k'}}$.
   Since this holds for each pair $(k,k') \in \setdef{(1,2), (1,3), (2,3)}$, the \almostDisjointTwoPathDecomposition $\X$ must induce a \almostDisjointTwoPathDecomposition of the subgraph $D_j \setdef{ (c_1,p_1), (p_1,c_2), (c_1,p_2), (p_2,c_2), (c_1,p_3), (p_3,c_2) }$ of the clause gadget $D_j$.
   By Proposition~\ref{TheoremClauseGadgetDecompositions}, this is impossible, which contradicts the assumption that $\clause_j$ was not satisfied.
\end{proof}

This equivalence implies that hardness results for directed and undirected graphs.
\begin{theorem}
   \label{TheoremHardness}
   The \probAlmostDisjointTwoPathDecomposition problem is $\cplxNP$-complete for directed and undirected graphs.
\end{theorem}

\begin{proof}
   The \probAlmostDisjointTwoPathDecomposition problem is in $\cplxNP$ as verifying that a given set $\X$ is indeed a \almostDisjointTwoPathDecomposition can be done in polynomial time by checking the definition of a \almostDisjointTwoPathDecomposition.

   Since the reduction described in this section can be carried out in polynomial time, the equivalence in Lemma~\ref{TheoremHardnessEquivalence} yields $\cplxNP$-hardness.
   This concludes the proof.
\end{proof}

We now show that an optimization version of the \probAlmostDisjointTwoPathDecomposition problem is not approximable on general (di)graphs.
Consider the following optimization problem of the \probAlmostDisjointTwoPathDecomposition problem, denoted by \probAlmostDisjointTwoPathDecomposition-MIN.
Given a (di)graph, partition the arcs (resp.\ edges) into paths of length 2 such that the number of paths that pairwise share more than one vertex is minimized.

\begin{corollary}
   \label{TheoremHardnessMinimization}
   \probAlmostDisjointTwoPathDecomposition-MIN is inapproximable for directed and undirected graphs unless $\cplxP = \cplxNP$.
\end{corollary}

\begin{proof}
    This follows directly from Theorem~\ref{TheoremHardness} by observing that any approximation algorithm would be able to solve the decision version of the problem since the (di)graph admits a \almostDisjointTwoPathDecomposition if and only if \probAlmostDisjointTwoPathDecomposition-MIN has optimum value $0$.
\end{proof}

%Now consider a different optimization problem: Given a digraph $D=(V,A)$, partition the arcs into paths of length 2 such that the number of path is maximized. 
%If $D$ admits a \almostDisjointTwoPathDecomposition, then it holds that \ptwo -MAX($D$) = $\frac{1}{2}m$. 

%% file: claw-free.tex
\section{Graph classes related to claw-free graphs}
\label{SectionClawFree}

In this section we characterize classes of graphs and digraphs for which the \probAlmostDisjointTwoPathDecomposition problem can be reduced to the stable-set problem in claw-free graphs.
For such (di)graphs, the problem can consequently be solved in polynomial time using Minty's algorithm~\cite{Minty80}.
The same reduction allows to solve the weighted version using the algorithm of Nakamura and Tamura~\cite{NakamuraT01}.
The reduction is the one mentioned in the introduction, for which we consider the \emph{conflict graph} $\conflictGraph$ of a graph.
As a reminder, for a graph $G$ (resp.\ digraph $D$), we denote by \twoPaths the set of all 2-paths in $G$ (resp.\ $D$).
The \emph{conflict graph} is defined as $\conflictGraph = (\twoPaths,\forbiddenPairs)$, where $\forbiddenPairs$ satisfies~\eqref{EquationForbiddenPairs}.

Note that the \probAlmostDisjointTwoPathDecomposition problem on $G$ (or $D$) can be reduced to the problem of deciding whether the cardinality maximum stable sets in $\conflictGraph$ have a certain size.
Since the maximum stable-set problem can be solved in polynomial time on claw-free graphs using Minty's algorithm~\cite{Minty80}, we conclude that the \probAlmostDisjointTwoPathDecomposition problem can also be solved in polynomial time if we restrict our set of instances to those having a claw-free conflict graph.
The remainder of this section is dedicated to the characterization of this restriction for the original (di)graph.

\paragraph{Undirected graphs.}
Consider an induced subgraph in the conflict graph $\conflictGraph$ that is a claw, as depicted in Figure~\ref{FigureClaw}.
Let $a$, $b$ and $c$ be the nodes of the path $P$ in $G$ that corresponds to the center node of the claw. 
We introduce the notation of $P = (a,b,c)$ as a short version of the undirected path $P = \setdef{\setdef{a,b}, \setdef{b,c}}$ and of the directed path $P = \setdef{(a,b),(b,c)}$.
Since the paths $Q$, $R$ and $S$ corresponding to the other nodes of the claw are not in conflict with each other,
 each of them must share exactly two nodes with $P$, but at most one node with each other.
Without loss of generality, we can assume $V(P) \cap V(Q) = \setdef{a,b}$, $V(P) \cap V(R) = \setdef{a,c}$ and $V(P) \cap V(S) = \setdef{b,c}$.

\begin{figure}[!ht]
   \begin{center}
      \resizebox{0.15\textwidth}{!}{
      \begin{tikzpicture}[
         align=center,
         node/.style = {circle,draw,fill,inner sep=0pt,minimum size=4pt,below=1cm, above=1cm, left=0.5cm, right=0.5cm},
         edge/.style = {very thick,below=1cm, above=1cm, left=0.5cm, right=0.5cm},
      ]
         \node[node,label=above right:{$P$}] (p) at (0,0) {};
         \node[node,label=above:{$Q$}] (q) at ({90+360/3 * (1 - 1)}:1cm) {};
         \draw[edge] (p) -- (q);

         \node[node,label=above:{$R$}] (r) at ({90+360/3 * (2 - 1)}:1cm) {};
         \draw[edge] (p) -- (r);

         \node[node,label=above:{$S$}] (s) at ({90+360/3 * (3 - 1)}:1cm) {};
         \draw[edge] (p) -- (s);
      \end{tikzpicture}
      } % resizebox
      \caption{A claw.}
      \label{FigureClaw}
   \end{center}
\end{figure}

This implies
\begin{alignat*}{6}
   Q &\in \setdef{ (x,a,b), (a,x,b), (a,b,x) } && \text{ for some $x \in V \setminus \setdef{c}$}, \\
   R &\in \setdef{ (y,a,c), (a,y,c), (a,c,y) } && \text{ for some $y \in V \setminus \setdef{b}$ and} \\
   S &\in \setdef{ (z,b,c), (b,z,c), (b,c,z) } && \text{ for some $z \in V \setminus \setdef{a}$},
\end{alignat*}
where $x$, $y$ and $z$ are pairwise distinct.
The subgraphs of $G$ corresponding to all possible combinations of the paths $P$, $Q$, $R$ and $S$ are depicted in Figure~\ref{fig:forbiddenSubgraphs}.

% if the following figures shall appear in separate figure environments, the following commands need to be defined before the first one
% p = (a,b,c)
% q = (x,a,b)
% r = (y,a,c)
% s = (z,b,c)

% define coordinates of all the nodes
\def \coordinateA {0,0}
\def \coordinateB {1,1}
\def \coordinateC {2,0}
\def \coordinateX {0,1}
\def \coordinateY {1,-1}
\def \coordinateZ {2,1}

% define nodes
\tikzset{node/.style = {circle,draw,fill,inner sep=0pt,minimum size=4pt,below=1cm, above=1cm, left=0.5cm, right=0.5cm}}
\newcommand{\nodeA}{\node[node,label=below:{$a$}] (A) at (\coordinateA) {}}
\newcommand{\nodeB}{\node[node,label=$b$] (B) at (\coordinateB) {}}
\newcommand{\nodeC}{\node[node,label=below:{$c$}] (C) at (\coordinateC) {}}
\newcommand{\nodeX}{\node[node,label=$x$] (X) at (\coordinateX) {}}
\newcommand{\nodeY}{\node[node,label=below:{$y$}] (Y) at (\coordinateY) {}}
\newcommand{\nodeZ}{\node[node,label=$z$] (Z) at (\coordinateZ) {}}

\newcommand{\drawNodes}{\nodeA; \nodeB; \nodeC; \nodeX; \nodeY; \nodeZ}

% define paths
\tikzset{drawStyleP/.style = {solid, very thick}}
\newcommand{\pathP}{\draw[drawStyleP] (A) -- (B) (B) -- (C)}

\tikzset{drawStyleQ/.style = {solid, very thick}}
\newcommand{\pathQ}[1]{
	\ifdefequal{#1}{1}{\draw[drawStyleQ] (X) --(A) (A) -- (B)}{
		\ifdefequal{#1}{2}{\draw[drawStyleQ] (A) --(X) (X) -- (B)}{\draw[drawStyleQ] (A) -- (B) (B) -- (X)}	
	}	
}
\tikzset{drawStyleR/.style = {solid, very thick}}
\newcommand{\pathR}[1]{
	\ifdefequal{#1}{1}{\draw[drawStyleR] (Y) --(A) (A) -- (C)}{
		\ifdefequal{#1}{2}{\draw[drawStyleR](A) --(Y) (Y) -- (C)}{\draw[drawStyleR] (A) --(C) (C) -- (Y)}	
	}
}
\tikzset{drawStyleS/.style = {solid, very thick}}
\newcommand{\pathS}[1]{
	\ifdefequal{#1}{1}{\draw[drawStyleS] (Z) --(B) (B) -- (C)}{
		\ifdefequal{#1}{2}{\draw[drawStyleS] (B) --(Z) (Z) -- (C)}{\draw[drawStyleS] (B) -- (C) (C) -- (Z)}	
	}
}

\newcounter{forbiddenSubgraph}
\setcounter{forbiddenSubgraph}{0}

\begin{figure}[!ht]
\begin{center}
\foreach \q in {1,2,3}
{
	\foreach \r in {1,2,3}
	{
		\foreach \s in {1,2,3}
		{
		\begin{subfigure}[t]{0.14\textwidth}
		\stepcounter{forbiddenSubgraph}
		\resizebox{\textwidth}{!}{
		\begin{tikzpicture}
			\drawNodes;
			\pathP;
			\ifthenelse{\equal{\q}{1}}{\pathQ{1};}{\ifthenelse{\equal{\q}{2}}{\pathQ{2};}{\pathQ{3};}}
			\ifthenelse{\equal{\r}{1}}{\pathR{1};}{\ifthenelse{\equal{\r}{2}}{\pathR{2};}{\pathR{3};}}
			\ifthenelse{\equal{\s}{1}}{\pathS{1};}{\ifthenelse{\equal{\s}{2}}{\pathS{2};}{\pathS{3};}}
		\end{tikzpicture}
		} %resizebox
		\subcaption{\ }
		\label{fig:forbiddenSubgraph\theforbiddenSubgraph}
		\end{subfigure}
		}
	}
}
\caption{Subgraphs resulting in a claw in the corresponding conflict graph $H$}
\label{fig:forbiddenSubgraphs}
\end{center}
\end{figure}
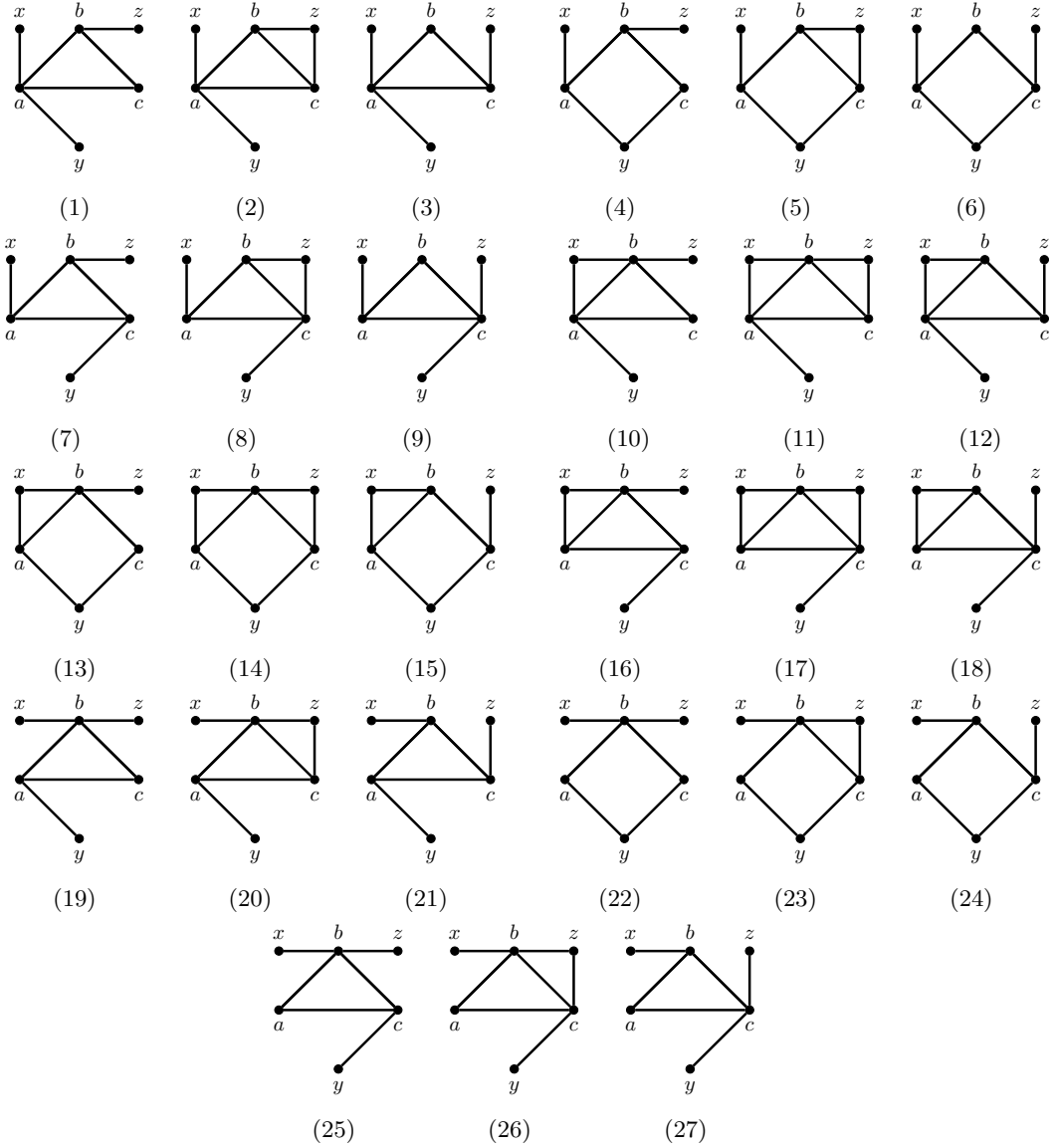

\DeclareDocumentCommand\clawfreeUndirectedForbidden{}{\ensuremath{\mathfrak{F}_{\text{und}}}\xspace}

Since all graphs in Figure~\ref{fig:forbiddenSubgraphs} contain Graph~\subref{fig:forbiddenSubgraph1}, \subref{fig:forbiddenSubgraph3},
\subref{fig:forbiddenSubgraph4}, \subref{fig:forbiddenSubgraph6}, \subref{fig:forbiddenSubgraph7}, \subref{fig:forbiddenSubgraph9},
\subref{fig:forbiddenSubgraph19} or~\subref{fig:forbiddenSubgraph22} as a subgraph,
the list of minimal forbidden subgraphs, denoted by \clawfreeUndirectedForbidden{}, contains exactly these seven graphs.

% if the following figures shall appear in separate figure environments, the following commands need to be defined before the first one
% p = (a,b,c)
% q = (x,a,b)
% r = (y,a,c)
% s = (z,b,c)

% define coordinates of all the nodes
\def \coordinateA {0,0}
\def \coordinateB {1,1}
\def \coordinateC {2,0}
\def \coordinateX {0,1}
\def \coordinateY {1,-1}
\def \coordinateZ {2,1}

% define nodes
\begin{comment}
\tikzset{node/.style = {circle,draw,fill,inner sep=0pt,minimum size=4pt,below=1cm, above=1cm, left=0.5cm, right=0.5cm}}
\newcommand{\nodeA}{\node[node,label=below:{$a$}] (A) at (\coordinateA) {}}
\newcommand{\nodeB}{\node[node,label=$b$] (B) at (\coordinateB) {}}
\newcommand{\nodeC}{\node[node,label=below:{$c$}] (C) at (\coordinateC) {}}
\newcommand{\nodeX}{\node[node,label=$x$] (X) at (\coordinateX) {}}
\newcommand{\nodeY}{\node[node,label=below:{$y$}] (Y) at (\coordinateY) {}}
\newcommand{\nodeZ}{\node[node,label=$z$] (Z) at (\coordinateZ) {}}

\newcommand{\drawNodes}{\nodeA; \nodeB; \nodeC; \nodeX; \nodeY; \nodeZ}

% define paths
\tikzset{drawStyleP/.style = {solid}}
\newcommand{\pathP}{\draw[drawStyleP] (A) -- (B) (B) -- (C)}

\tikzset{drawStyleQ/.style = {solid}}
\newcommand{\pathQ}[1]{
	\ifdefequal{#1}{1}{\draw[drawStyleQ] (X) --(A) (A) -- (B)}{
		\ifdefequal{#1}{2}{\draw[drawStyleQ] (A) --(X) (X) -- (B)}{\draw[drawStyleQ] (A) -- (B) (B) -- (X)}	
	}	
}
\tikzset{drawStyleR/.style = {solid}}
\newcommand{\pathR}[1]{
	\ifdefequal{#1}{1}{\draw[drawStyleR] (Y) --(A) (A) -- (C)}{
		\ifdefequal{#1}{2}{\draw[drawStyleR](A) --(Y) (Y) -- (C)}{\draw[drawStyleR] (A) --(C) (C) -- (Y)}	
	}
}
\tikzset{drawStyleS/.style = {solid}}
\newcommand{\pathS}[1]{
	\ifdefequal{#1}{1}{\draw[drawStyleS] (Z) --(B) (B) -- (C)}{
		\ifdefequal{#1}{2}{\draw[drawStyleS] (B) --(Z) (Z) -- (C)}{\draw[drawStyleS] (B) -- (C) (C) -- (Z)}	
	}
}
\end{comment}

\tikzset{drawStyleQAdditional/.style = {very thick,dashed}}
\newcommand{\pathQAdd}[1]{
	\ifdefequal{#1}{1}{\draw[drawStyleQAdditional] (X) -- (B)}{
		\draw[drawStyleQAdditional] (X) -- (A)}
}

\tikzset{drawStyleSAdditional/.style = {very thick,dotted}}
\newcommand{\pathSAdd}[1]{
	\ifdefequal{#1}{1}{\draw[drawStyleSAdditional] (Z) -- (C)}{
		\draw[drawStyleSAdditional] (Z) -- (B)}
}

\newcounter{forbiddenSubgraphUndirected}
\setcounter{forbiddenSubgraphUndirected}{0}

\begin{figure}[!ht]
\begin{center}
\foreach \q in {1,3}
{
	\foreach \r in {1,2}
	{
		\foreach \s in {1,3}
		{
		\ifthenelse{\equal{\q}{3} \AND \equal{\r}{2} \AND \equal{\s}{3}}{\ }{%the combination of q,r,s=3,2,3 produces a graph isomorphic to 3,2,1 
		\begin{subfigure}[t]{0.14\textwidth}
		\stepcounter{forbiddenSubgraphUndirected}
		\resizebox{\textwidth}{!}{
		\begin{tikzpicture}
			\drawNodes;
			\pathP;
			\ifthenelse{\equal{\q}{1}}{\pathQ{1};}{\ifthenelse{\equal{\q}{2}}{\pathQ{2};}{\pathQ{3};}}
			%\ifthenelse{\equal{\q}{1}}{\pathQAdd{1};}{\pathQAdd{2};}
			\ifthenelse{\equal{\r}{1}}{\pathR{1};}{\ifthenelse{\equal{\r}{2}}{\pathR{2};}{\pathR{3};}}
			%\ifthenelse{\equal{\s}{1}}{\pathSAdd{1};}{\pathSAdd{2};}
			\ifthenelse{\equal{\s}{1}}{\pathS{1};}{\ifthenelse{\equal{\s}{2}}{\pathS{2};}{\pathS{3};}}
		\end{tikzpicture}
		} %resizebox
		\subcaption{\ }
		\label{fig:forbiddenSubgraphUndirected\theforbiddenSubgraphUndirected}
		\end{subfigure}
		}
		}
	}
	
}
\caption{Forbidden sugraphs in undirected graphs.}
\label{fig:fewerForbiddenSubgraphsUndirected}
\end{center}
\end{figure}

\begin{comment}
\begin{figure}[!ht]
\begin{center}
\begin{subfigure}[t]{0.14\textwidth}
\resizebox{\textwidth}{!}{
\begin{tikzpicture}
	\drawNodes;
	\pathP;
	\pathQ{2};
	\pathR{1};
	\pathS{2};
\end{tikzpicture}
} %resizebox
\subcaption{\ }
\end{subfigure}
%
\begin{subfigure}[t]{0.14\textwidth}
\resizebox{\textwidth}{!}{
\begin{tikzpicture}
	\drawNodes;
	\pathP;
	\pathQ{2};
	\pathR{2};
	\pathS{2};
\end{tikzpicture}
} %resizebox
\subcaption{\ }
\end{subfigure}
\caption{Forbidden sugraphs in undirected graphs}
\label{fig:fewerForbiddenSubgraphsUndirected}
\end{center}
\end{figure}
\end{comment}

This proves the following structural result for undirected graphs.
\begin{theorem}
   \label{thm:charForbiddenSubgraphsUndirected}
   Let $G$ be a graph and let $\conflictGraph$ be its conflict graph.
   Then $\conflictGraph$ contains a claw if and only if $G$ contains one of the graphs from $\clawfreeUndirectedForbidden$ as a subgraph.
\end{theorem}

We immediately obtain the following complexity result for undirected graphs.
\begin{corollary}
   \label{thm:cplxForbiddenSubgraphsUndirected}  
   The (weighted) \probAlmostDisjointTwoPathDecomposition problem for undirected graphs that do not contain a subgraph from $\clawfreeUndirectedForbidden$ can be solved in polynomial time.
\end{corollary}

\begin{proof}
   Let $G$ be a graph that does not contain any subgraph in $\clawfreeUndirectedForbidden$.
   We construct the conflict graph $\conflictGraph$ of $G$ in polynomial time.
   By Theorem~\ref{thm:charForbiddenSubgraphsUndirected}, $\conflictGraph$ is claw-free.
   The (weighted) \probAlmostDisjointTwoPathDecomposition-problem on $G$ is therefore equivalent to solving the maximum (weight) stable-set problem on $\conflictGraph$.
   Using Minty's Algorithm~\cite{Minty80} we can find a stable set of size $|E|/2$ if such a set exists.
   In the affirmative case, we can derive from a given objective vector $c \in \R^{\twoPaths}$ a vector $c' \in \R^{\twoPaths}$ such that the $c$-maximum stable sets of size $|E|/2$ are exactly the $c'$-maximum stable sets.
   Using the algorithm by Nakamura and Tamura~\cite{NakamuraT01} we can then find such a set and derive the corresponding \almostDisjointTwoPathDecomposition from it.
\end{proof}

\paragraph{Directed graphs.}

We now turn to the directed case.
The following lemma shows that a claw in the conflict graph $\conflictGraph$ of $D$ induces a claw in the conflict graph $\conflictGraph'$ of the underlying undirected multi-graph $G(D)$.

\begin{lemma}
   \label{TheoremSubgraphsYieldClaw}
   Let $D$ be a digraph, let $\conflictGraph$ be its conflict graph, and let $\conflictGraph'$ be the conflict graph of $G(D)$.
   Then each induced claw in $\conflictGraph$ induces a claw in $\conflictGraph'$.
\end{lemma}

\begin{proof}
   Let $\conflictGraph = (\twoPaths,\forbiddenPairs)$ and $\conflictGraph' = (\twoPaths',\forbiddenPairs')$ and consider two paths $P, Q \in \twoPaths$ with $P = \setdef{(a,b),(b,c)}$ and $Q = \setdef{(x,y),(y,z)}$. 
   Define the two undirected paths $P' := \setdef{\setdef{a,b},\setdef{b,c}}$ and $Q' := \setdef{\setdef{x,y},\setdef{y,z}}$.
   Since being in conflict is a property of the nodes of a path, $P$ and $Q$ are in conflict in $D$ if and only if $P'$ and $Q'$ are in conflict in $G(D)$.
   Thus, a claw in $\conflictGraph$ having four $2$-paths in $D$ as its nodes induces a claw in $\conflictGraph'$ having the corresponding undirected $2$-paths in $G(D)$ as its nodes, which proves the claim.
\end{proof}

Due to Lemma~\ref{TheoremSubgraphsYieldClaw} we restrict our analysis to graphs that are directed versions of the graphs depicted in Figure~\ref{fig:forbiddenSubgraphs}. 
Furthermore, out of these graphs we are only interested in graphs that result in a claw in $\conflictGraph$.
In particular, those edges that are combined to paths $Q$, $R$ and $S$ have to be directed such that the directed counterparts also form paths.
Therefore, we are not interested in all possible directed versions of the depicted graphs. 
For example, consider the Graph~\subref{fig:forbiddenSubgraph1} in Figure~\ref{fig:forbiddenSubgraphs} and assume for a directed version $D = (V,A)$ of it, $(a,x),(a,b) \in A$.
Then neither $Q = (x,a,b)$ is a path in $D$ nor $Q' = (b,a,x)$ and therefore $Q, Q' \notin \twoPaths$.
Note that the node corresponding to the (undirected) path $(x,a,b)$ in $\conflictGraph'$ is part of a claw. 
Furthermore, note that as $Q, Q' \notin \twoPaths$ there is no corresponding claw in $\conflictGraph$.
Hence, we are interested in directed versions of graphs depicted in Figure~\ref{fig:forbiddenSubgraphs} that contain the directed paths $Q$, $R$ and $S$ with
\begin{alignat*}{6}
  Q &\in \setdef{(x,a,b),(b,a,x),(a,x,b),(b,x,a),(a,b,x),(x,b,a)}, \\
  R &\in \setdef{(y,a,c),(c,a,y),(a,y,c),(c,y,a),(a,c,y),(y,c,a)} \text{ and} \\
  S &\in \setdef{(z,b,c),(c,b,y),(b,z,c), (c,z,b), (b,c,z),(z,b,c)}.
\end{alignat*}
These digraphs are depicted in Figure~\ref{fig:fewerForbiddenSubgraphsDirected}. 

% if the following figures shall appear in separate figure environments, the following commands need to be defined before the first one
% p = (a,b,c)
% q = (x,a,b)
% r = (y,a,c)
% s = (z,b,c)

% define coordinates of all the nodes
\def \coordinateA {0,0}
\def \coordinateB {1,1}
\def \coordinateC {2,0}
\def \coordinateX {0,1}
\def \coordinateY {1,-1}
\def \coordinateZ {2,1}

% define nodes
\tikzset{node/.style = {circle,draw,fill,inner sep=0pt,minimum size=4pt,below=1cm, above=1cm, left=0.5cm, right=0.5cm}}
%\newcommand{\nodeA}{\node[node,label=below:{$a$}] (A) at (\coordinateA) {}}
%\newcommand{\nodeB}{\node[node,label=$b$] (B) at (\coordinateB) {}}
%\newcommand{\nodeC}{\node[node,label=below:{$c$}] (C) at (\coordinateC) {}}
%\newcommand{\nodeX}{\node[node,label=$x$] (X) at (\coordinateX) {}}
%\newcommand{\nodeY}{\node[node,label=below:{$y$}] (Y) at (\coordinateY) {}}
%\newcommand{\nodeZ}{\node[node,label=$z$] (Z) at (\coordinateZ) {}}

%\newcommand{\drawNodes}{\nodeA; \nodeB; \nodeC; \nodeX; \nodeY; \nodeZ}
% define paths
\tikzset{drawStyleP/.style = {solid,very thick,->}}
\newcommand{\diPathP}{\draw[drawStyleP] (A) -- (B); \draw[drawStyleP] (B) -- (C)}

\tikzset{drawStyleQ/.style = {solid,very thick,->}}
\newcommand{\diPathQ}[2]{
% 1st parameter describes #arcs; 2nd parameter describes arc direction (1: forward; 2: backwards)
	\ifdefequal{#1}{1}{
		\ifdefequal{#2}{1}{
			%input is {1}{1}
			\draw[drawStyleQ] (X) --(A)}{
			%input is {1}{2}
			\draw[drawStyleQ] (B) -- (X)}}{
	\ifdefequal{#1}{2}{
		\ifdefequal{#2}{1}{
			%input is {2}{1}
			\draw[drawStyleQ] (A) --(X); \draw[drawStyleQ] (X) -- (B)}{
			%input is {2}{2}
			\draw[drawStyleQ] (B) --(X); \draw[drawStyleQ] (X) -- (A)}}{
	\ifdefequal{#2}{1}{
		%input is {3}{1}
		\draw[drawStyleQ] (B) -- (X)}{
		%input is {3}{1}
		\draw[drawStyleQ] (X) -- (B)}}
	}	
}

\tikzset{drawStyleR/.style = {solid,very thick,->}}
\newcommand{\diPathR}[2]{
% 1st parameter describes node order (1: y,a,c; 2: a,y,c; 3: a,c,y); 2nd parameter describes arc direction (1: forward; 2: backwards)
	\ifdefequal{#1}{1}{
		\ifdefequal{#2}{1}{
			%input is {1}{1}
			\draw[drawStyleR] (Y) --(A); \draw[drawStyleR] (A) -- (C)}{
			%input is {1}{2}
			\draw[drawStyleR] (C) -- (A); \draw[drawStyleR] (A) -- (Y)}}{
	\ifdefequal{#1}{2}{
		\ifdefequal{#2}{1}{
			%input is {2}{1}
			\draw[drawStyleR] (A) --(Y); \draw[drawStyleR] (Y) -- (C)}{
			%input is {2}{2}
			\draw[drawStyleR] (C) --(Y); \draw[drawStyleR] (Y) -- (A)}}{
	\ifdefequal{#2}{1}{
		%input is {3}{1}
		\draw[drawStyleR] (A) --  (C); \draw[drawStyleR] (C) -- (Y)}{
		%input is {3}{1}
		\draw[drawStyleR] (Y) -- (C); \draw[drawStyleR] (C) -- (A)}}
	}	
}

\tikzset{drawStyleS/.style = {solid,very thick,->}}
\newcommand{\diPathS}[2]{
% 1st parameter describes #arcs; 2nd parameter describes arc direction (1: forward; 2: backwards)
	\ifdefequal{#1}{1}{
		\ifdefequal{#2}{1}{
			%input is {1}{1}
			\draw[drawStyleS] (Z) --(B)}{
			%input is {1}{2}
			\draw[drawStyleS] (C) -- (Z)}}{
	\ifdefequal{#1}{2}{
		\ifdefequal{#2}{1}{
			%input is {2}{1}
			\draw[drawStyleS] (B) --(Z); \draw[drawStyleS] (Z) -- (C)}{
			%input is {2}{2}
			\draw[drawStyleS] (C) --(Z); \draw[drawStyleS] (Z) -- (B)}}{
	\ifdefequal{#2}{1}{
		%input is {3}{1}
		\draw[drawStyleS] (C) -- (Z)}{
		%input is {3}{1}
		\draw[drawStyleS] (Z) -- (C)}}
	}	
}

\setcounter{forbiddenSubgraph}{0}

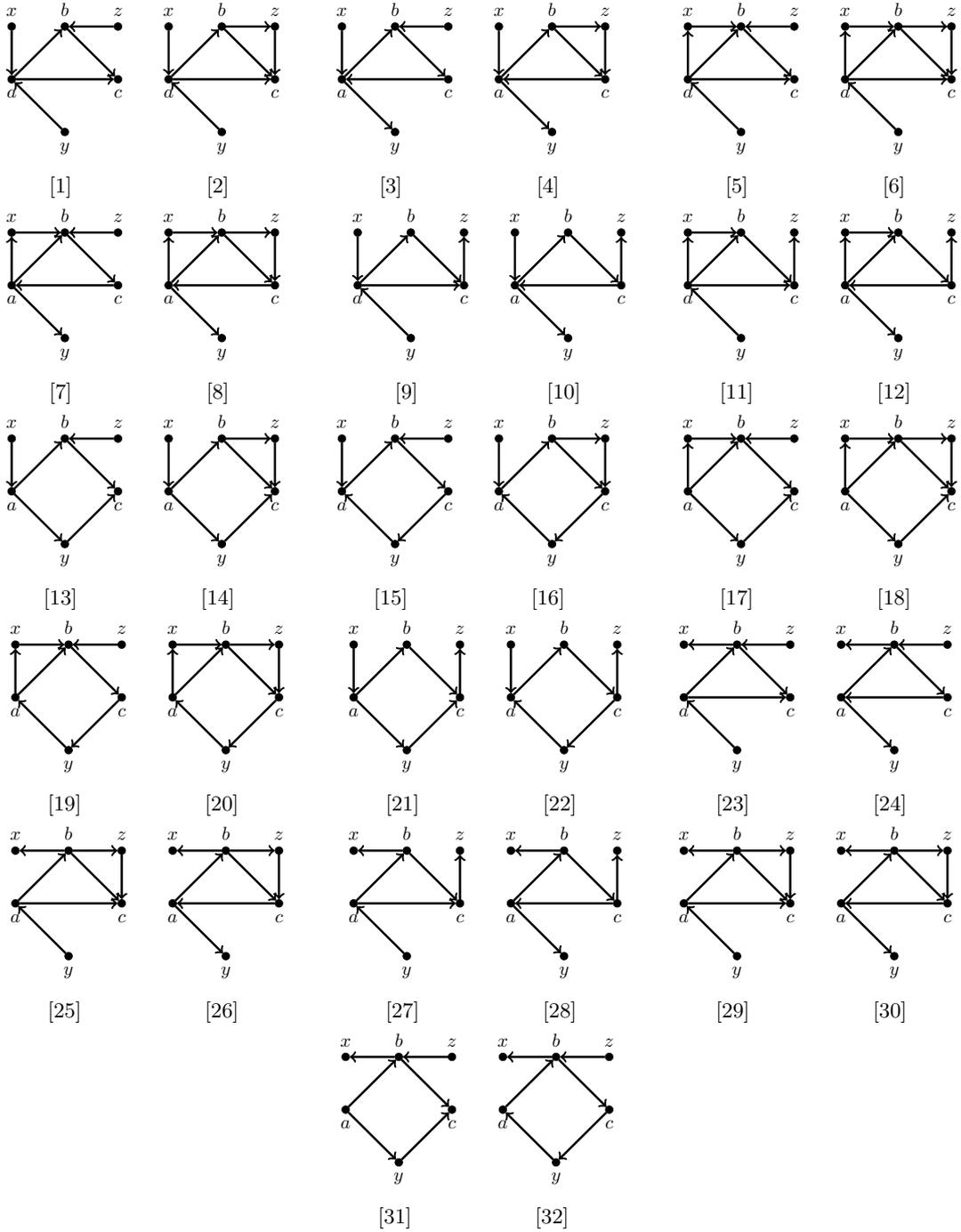
\begin{figure}[!ht]
   \renewcommand\thesubfigure{[\arabic{subfigure}]}
\begin{center}
	\foreach \q in {1,2}
	{
	\foreach \r in {1,2}
	{
		\foreach \s in {1,2}
		{
		\begin{subfigure}[t]{0.14\textwidth}
		\stepcounter{forbiddenSubgraph}
		\resizebox{\textwidth}{!}{
		\begin{tikzpicture}
			\drawNodes;
			\diPathP;
			\ifthenelse{\equal{\q}{1}}{\diPathQ{1}{1};}{\diPathQ{2}{1};}
			\ifthenelse{\equal{\r}{1}}{\diPathR{1}{1};}{\diPathR{1}{2};}
			\ifthenelse{\equal{\s}{1}}{\diPathS{1}{1};}{\diPathS{2}{1};}
		\end{tikzpicture}
		} %resizebox
		\subcaption{\ }
		\label{fig:forbiddenSubgraphDirected\theforbiddenSubgraph}
		\end{subfigure}
		}
	}
	}
	\foreach \q in {1,2}
	{
	\foreach \r in {1,2}
	{
		\begin{subfigure}[t]{0.14\textwidth}
		\stepcounter{forbiddenSubgraph}
		\resizebox{\textwidth}{!}{
		\begin{tikzpicture}
			\drawNodes;
			\diPathP;
			\ifthenelse{\equal{\q}{1}}{\diPathQ{1}{1};}{\diPathQ{2}{1};}
			\ifthenelse{\equal{\r}{1}}{\diPathR{1}{1};}{\diPathR{1}{2};}
			\diPathS{1}{2};
		\end{tikzpicture}
		} %resizebox
		\subcaption{\ }
		\label{fig:forbiddenSubgraphDirected\theforbiddenSubgraph}
		\end{subfigure}
	}
	}
	\foreach \q in {1,2}
	{
	\foreach \r in {1,2}
	{
		\foreach \s in {1,2}
		{
		\begin{subfigure}[t]{0.14\textwidth}
		\stepcounter{forbiddenSubgraph}
		\resizebox{\textwidth}{!}{
		\begin{tikzpicture}
			\drawNodes;
			\diPathP;
			\ifthenelse{\equal{\q}{1}}{\diPathQ{1}{1};}{\diPathQ{2}{1};}
			\ifthenelse{\equal{\r}{1}}{\diPathR{2}{1};}{\diPathR{2}{2};}
			\ifthenelse{\equal{\s}{1}}{\diPathS{1}{1};}{\diPathS{2}{1};}
		\end{tikzpicture}
		} %resizebox
		\subcaption{\ }
		\label{fig:forbiddenSubgraphDirected\theforbiddenSubgraph}
		\end{subfigure}
		}
	}
	}
	\foreach \r in {1,2}
	{
		\begin{subfigure}[t]{0.14\textwidth}
		\stepcounter{forbiddenSubgraph}
		\resizebox{\textwidth}{!}{
		\begin{tikzpicture}
			\drawNodes;
			\diPathP;
			\diPathQ{1}{1};
			\ifthenelse{\equal{\r}{1}}{\diPathR{2}{1};}{\diPathR{2}{2};}
			\diPathS{1}{2};
		\end{tikzpicture}
		} %resizebox
		\subcaption{\ }
		\label{fig:forbiddenSubgraphDirected\theforbiddenSubgraph}
		\end{subfigure}
	}
	\foreach \s in {1,2}
	{
	\foreach \r in {1,2}
	{
		\begin{subfigure}[t]{0.14\textwidth}
		\stepcounter{forbiddenSubgraph}
		\resizebox{\textwidth}{!}{
		\begin{tikzpicture}
			\drawNodes;
			\diPathP;
			\diPathQ{1}{2};
			\ifthenelse{\equal{\r}{1}}{\diPathR{1}{1};}{\diPathR{1}{2};}
			\ifthenelse{\equal{\s}{1}}{\diPathS{1}{1};}{\diPathS{2}{1};}
		\end{tikzpicture}
		} %resizebox
		\subcaption{\ }
		\label{fig:forbiddenSubgraphDirected\theforbiddenSubgraph}
		\end{subfigure}
	}
	}
	\foreach \s in {1,2}
	{
	\foreach \r in {1,2}
	{
		\begin{subfigure}[t]{0.14\textwidth}
		\stepcounter{forbiddenSubgraph}
		\resizebox{\textwidth}{!}{
		\begin{tikzpicture}
			\drawNodes;
			\diPathP;
			\diPathQ{1}{2};
			\ifthenelse{\equal{\r}{1}}{\diPathR{1}{1};}{\diPathR{1}{2};}
			\ifthenelse{\equal{\s}{1}}{\diPathS{1}{2};}{\diPathS{2}{1};}
		\end{tikzpicture}
		} %resizebox
		\subcaption{\ }
		\label{fig:forbiddenSubgraphDirected\theforbiddenSubgraph}
		\end{subfigure}
	}
	}
	\foreach \r in {1,2}
	{
		\begin{subfigure}[t]{0.14\textwidth}
		\stepcounter{forbiddenSubgraph}
		\resizebox{\textwidth}{!}{
		\begin{tikzpicture}
			\drawNodes;
			\diPathP;
			\diPathQ{1}{2};
			\ifthenelse{\equal{\r}{1}}{\diPathR{2}{1};}{\diPathR{2}{2};}
			\diPathS{1}{1};
		\end{tikzpicture}
		} %resizebox
		\subcaption{\ }
		\label{fig:forbiddenSubgraphDirected\theforbiddenSubgraph}
		\end{subfigure}
	}
	
\caption{Forbidden subgraphs in directed graphs.}
\label{fig:fewerForbiddenSubgraphsDirected}
\end{center}
\end{figure}

Table~\ref{TableDirectedSubgraphs} lists all directed versions of interest for each undirected graph in Figure~\ref{fig:forbiddenSubgraphs}.

\DeclareDocumentCommand\sg{m}{[#1]\,$\leq$\xspace}
\DeclareDocumentCommand\sgiso{m}{[#1]\,$\lessapprox$\xspace}
\DeclareDocumentCommand\iso{m}{[#1]\,$\cong$\xspace}
\DeclareDocumentCommand\id{m}{[#1]\,$=$\xspace}

\begin{table}[!ht]
   \begin{center}
      \begin{tabular}{r|l}
         \textbf{Graph} & \textbf{Directed versions inducing a claw} \\ \hline
         (1) & \id{1}, \id{3} \\
         (2) & \sg{1}, \sg{3}, \id{2}, \id{4} \\
         (3) & \id{9}, \id{10} \\
         (4) & \id{13}, \id{15} \\
         (5) & \sg{13}, \id{14}, \sg{15}, \id{16} \\
         (6) & \id{21}, \id{22}\\
         (7) & \iso{29}, \iso{30} \\
         (8) & \sgiso{6}, \sgiso{8}, \iso{11}, \iso{12} \\
         (9) & \iso{9}, \iso{10} \\
         (10) & \sg{1}, \sg{3}, \id{5}, \id{7} \\
         (11) & \sg{2}, \id{6}, \id{8}, \sg{9}, \sg{10}, \sg{11} \\
         (12) & \sg{6}, \sg{8}, \id{11}, \id{12} \\
         (13) & \sg{13}, \sg{15}, \id{17}, \id{19} \\
         (14) & \sg{14}, \sg{16}, \id{18}, \id{20}, \sg{23}, \sg{24} \\
         (15) & \sgiso{13}, \iso{14}, \sgiso{15}, \iso{16} \\
         (16) & \sgiso{25}, \sgiso{26}, \iso{27}, \iso{28} \\
         (17) & \sgiso{2}, \iso{6}, \iso{8}, \sgiso{9}, \sgiso{10}, \sgiso{11} \\
         (18) & \sgiso{1}, \sgiso{3}, \iso{2}, \iso{4} \\
         (19) & \iso{25}, \iso{26} \\
         (20) & \sg{25}, \sg{26}, \id{27}, \id{28} \\
         (21) & \id{29}, \id{30} \\
         (22) & \id{33}, \id{34} \\
         (23) & \sgiso{13}, \sgiso{15}, \iso{17}, \iso{19} \\
         (24) & \iso{13}, \iso{15} \\
         (25) & \iso{25}, \iso{26} \\
         (26) & \sgiso{25}, \sgiso{26}, \iso{27}, \iso{28} \\
         (27) & \iso{1}, \iso{3}
      \end{tabular}
   \end{center}
   \caption{All directed versions of the undirected graphs in Figure~\ref{fig:forbiddenSubgraphs} that induce the same claw in the conflict graph.
   The symbol after each digraph denotes whether the digraph's underlying undirected graph is equal to ($=$), a subgraph of ($\leq$), isomorphic to ($\cong$) or isomorphic to a subgraph of ($\lessapprox$) the undirected graph in the first column.}
   \label{TableDirectedSubgraphs}
\end{table}

\DeclareDocumentCommand\clawfreeDirectedForbidden{}{\ensuremath{\mathfrak{F}_{\text{dir}}}\xspace}

We therefore define $\clawfreeDirectedForbidden$ as the set containing the graphs represented in Figure~\ref{fig:fewerForbiddenSubgraphsDirected} as well as those arising by reverting all arc directions and directly obtain the following directed counterpart to Theorem~\ref{thm:charForbiddenSubgraphsUndirected}.
\begin{theorem}
   \label{thm:charForbiddenSubgraphsDirected}
   Let $D$ be a digraph and let $\conflictGraph$ be its conflict graph.
   Then $\conflictGraph$ contains a claw if and only if $D$ contains one of the graphs in $\clawfreeDirectedForbidden$ as a subgraph.
\end{theorem}

The corresponding complexity results reads as follows.
\begin{corollary}
  \label{thm:cplxForbiddenSubgraphsDirected}  
  The (weighted) \probAlmostDisjointTwoPathDecomposition problem for digraphs that do not contain a subgraph from $\clawfreeDirectedForbidden$ can be solved in polynomial time.
\end{corollary}

The proof is omitted since it is almost identical to that of Corollary~\ref{thm:cplxForbiddenSubgraphsUndirected}.
A simple observation is that all (di)graphs in $\clawfreeUndirectedForbidden{}$ and $\clawfreeDirectedForbidden$ contain a cycle with an adjacent edge, which yields the following consequence.
\begin{corollary}
  The \probAlmostDisjointTwoPathDecomposition problem on paths, trees and cycles can be solved in polynomial time.
\end{corollary}

\clearpage

%% file: series-parallel.tex
\section{Series-parallel digraphs}
\label{SectionSeriesParallel}

We now consider the graph class of series-parallel digraphs.
We propose a polynomial dynamic program that solves the \probAlmostDisjointTwoPathDecomposition problem on such digraphs.
First, we give a formal definition of series-parallel (SP) digraphs and then describe the dynamic program in detail.  

\begin{definition}[Series-parallel digraph]
   A digraph $D$ is called \emph{series-parallel} (\emph{SP-digraph}) if it has a single arc or is composed of two smaller series-parallel digraphs $D_1$ and $D_2$, both having unique sources $s_i$ and unique sinks $t_i$ for $i=1,2$, in one of the following ways:
   \begin{enumerate}
   \item
      \emph{Series composition:}
      $D$ is obtained from $D_1$ and $D_2$ by considering their disjoint union and identifying
      $t_2$ with $s_1$.
   \item
      \emph{Parallel composition:}
      $D$ is obtained from $D_1$ and $D_2$ by considering their disjoint union and identifying
      $s_1$ with $s_2$ and $t_1$ with $t_2$.
   \end{enumerate}
\end{definition}

\DeclareDocumentCommand\N{}{\mathbb{N}}

\DeclareDocumentCommand\outArcs{om}{\IfValueTF{#1}{\delta_{#1}^{\text{out}}(#2)}{\delta^{\text{out}}(#2)}}
\DeclareDocumentCommand\inArcs{om}{\IfValueTF{#1}{\delta_{#1}^{\text{in}}(#2)}{\delta^{\text{in}}(#2)}}
\DeclareDocumentCommand\subgraph{o}{\IfValueTF{#1}{D_{#1} \setminus (S_{#1} \cup T_{#1})}{D \setminus (S \cup T)}}

\begin{definition}[feasible configurations for an SP-digraph]
   \label{DefinitionFeasibleConfiguration}
   Let $D$ be a simple SP-digraph with source $s$ and sink $t$.
   A triple $(a,b,c) \in \N \times \N \times \setdef{0,1,2}$ is called a \emph{$D$-feasible configuration} 
   if there exist subsets $S \subseteq \outArcs[D]{s}$, $T \subseteq \inArcs[D]{t}$
   and a \almostDisjointTwoPathDecomposition $\X$ of $\subgraph$ with the following properties:
   \begin{enumerate}[label=(\roman*)]
   \item
      \label{FeasibleConfigurationCardinalities}
      \label{FeasibleConfigurationFirst}
      $a = |S|$ and $b = |T|$.
   \item
      \label{FeasibleConfigurationTriangles}
      For no arc $(u,v) \in S \cup T$ there is a $u$-$v$-path in $\X$.
   \item
      \label{FeasibleConfiguration1Path}
      $D$ contains the arc $(s,t)$ if and only if $c = 1$. In this case, $(s,t) \in S \cap T$.
   \item
      \label{FeasibleConfiguration2Path}
      \label{FeasibleConfigurationLast}
      $\X$ contains an $s$-$t$-path (of length $2$) if and only if $c = 2$.
   \end{enumerate}
\end{definition}

In Definition~\ref{DefinitionFeasibleConfiguration}, the digraph $D = (V,A)$ may have only a partial decomposition into $2$-paths.
We say that an arc $a \in A$ is \emph{covered} if $a \in P$ for some path $P \in \X$ and that $a$ is \emph{free} otherwise.
Note that $c$ in Definition~\ref{DefinitionFeasibleConfiguration} is well-defined by Property~\ref{FeasibleConfigurationTriangles}.
Furthermore, $D$ has a \almostDisjointTwoPathDecomposition if and only if $(0,0,c)$ is $D$-feasible for some $c \in \setdef{0,2}$.

We will now characterize how feasible configurations of a series composition are related to feasible configurations of the two component graphs $D_i$ with sources $s_i$ and sinks $t_i$ (for $i=1,2$),
identifying $t_1$ with $s_2$.
By distinguishing the cases, for each $i=1,2$, whether in the decomposition of the composed graph, an $s_i$-$t_i$-arc is free (F), is covered by a path that uses edges from both graphs (C), or does not exist (E), we obtain several relations, which will be justified in a subsequent lemma:
\begin{itemize}[label=\arabic]
\item[\mylabel{FeasibleConfigurationSeriesEmptyEmpty}{(EE)}] % U_1 empty, U_2 empty
   $(a,b,c) = (a_1, b_2, 0)$, $b_1 = a_2$ and $c_1,c_2\neq 1$ 
\item[\mylabel{FeasibleConfigurationSeriesFreeEmpty}{(FE)}] % U_1 free, U_2 empty
   $(a,b,c) = (a_1, b_2, 0)$, $b_1 = a_2 + 1$, $c_1 = 1$ and $c_2 \neq 1$ 
\item[\mylabel{FeasibleConfigurationSeriesEmptyFree}{(EF)}] % U_1 empty, U_2 free
   $(a,b,c) = (a_1, b_2, 0)$, $b_1 = a_2-1$, $c_1 \neq 1$ and $c_2= 1$
\item[\mylabel{FeasibleConfigurationSeriesFreeFree}{(FF)}] % U_1 free, U_2 free
   $(a,b,c) = (a_1, b_2, 0)$, $b_1 = a_2$ and $c_1 = c_2 = 1$
\item[\mylabel{FeasibleConfigurationSeriesCoveredEmpty}{(CE)}] % U_1 covered, U_2 empty
   $(a,b,c) = (a_1 - 1, b_2, 0)$, $b_1 = a_2$, $c_1 = 1$ and $c_2 \neq 1$
\item[\mylabel{FeasibleConfigurationSeriesCoveredFree}{(CF)}] % U_1 covered, U_2 free
   $(a,b,c) = (a_1 - 1, b_2, 0)$, $b_1 = a_2 - 1$ and $c_1 = c_2 = 1$
\item[\mylabel{FeasibleConfigurationSeriesEmptyCovered}{(EC)}] % U_1 empty, U_2 covered
   $(a,b,c) = (a_1, b_2 - 1, 0)$, $b_1 = a_2$, $c_1 \neq 1$ and $c_2 = 1$
\item[\mylabel{FeasibleConfigurationSeriesFreeCovered}{(FC)}] % U_1 free, U_2 covered
   $(a,b,c) = (a_1, b_2 - 1, 0)$, $b_1=a_2+1$ and $c_1 = c_2= 1$
\item[\mylabel{FeasibleConfigurationSeriesCoveredCoveredUncombined}{(CC)}] % U_1 covered, U_2 covered, not combined to s1-t2-path
   $(a,b,c) = (a_1-1, b_2-1, 0)$, $b_1 = a_2 > 1$ and $c_1 = c_2 = 1$
\item[\mylabel{FeasibleConfigurationSeriesCoveredCoveredCombined}{(CC')}] % U_1 covered, U_2 covered, combined to s1-t2-path
   $(a,b,c) = (a_1-1, b_2-1, 2)$, $b_1=a_2$ and $c_1 = c_2 = 1$ 
\end{itemize}
   
\begin{lemma}[feasible configurations for series composition]
   \label{TheoremFeasibleConfigurationsSeries}
   Let $D$ be a simple SP-digraph with source $s$ and sink $t$ such that $D$ is obtained from SP-digraphs $D_i = (V_i, A_i)$ with sources $s_i$ and sinks $t_i$ for $i=1,2$ by a series operation, identifying $t_2$ with $s_1$. 
   Then a configuration $(a,b,c)$ is $D$-feasible if and only if there exist $D_i$-feasible configurations $(a_i,b_i,c_i)$ for $i=1,2$ that satisfy one of the constraints~\ref{FeasibleConfigurationSeriesEmptyEmpty}--\ref{FeasibleConfigurationSeriesCoveredCoveredCombined}.
\end{lemma}

\begin{proof}
   To see necessity, let $(a,b,c)$ be a $D$-feasible configuration.
   By definition, there exist sets $S \subseteq \outArcs[D]{s}$, $T \subseteq \inArcs[D]{t}$ and a \almostDisjointTwoPathDecomposition $\X$ of $\subgraph$ that satify Properties~\ref{FeasibleConfigurationFirst}--\ref{FeasibleConfigurationLast} from Definition~\ref{DefinitionFeasibleConfiguration}.
   By construction, for $i=1,2$, $\X_i := \setdef{ P \in \X }[ P \subseteq A_i ]$ is a \almostDisjointTwoPathDecomposition of $\subgraph[i]$ with $S_i = \setdef{ a \in \outArcs[D_i]{s_i} }[ a \notin A(\X_i) ]$ and $T_i = \setdef{ a \in \inArcs[D_i]{t_i} }[ a \notin A(\X_i)]$.
   Let $Y \subseteq A$ be the set of arcs belonging to a path of $\X$ that contains exactly one arc in $A_1$ (and the other one in $A_2$).
   Note that $Y \subseteq A \setminus (S \cup T)$.
   For $i = 1,2$, let $U_i := \setdef{ (s_i,t_i) \in A_i }$ be the set containing the arc $(s_i,t_i)$ if such an arc exists.
   Note that $U_i \subseteq S \cup T \cup Y$ since the arc cannot be part of a path in $D_i$.
   Let $S_1 := S \cup (U_1 \cap Y)$ and $T_1 := (\inArcs[D_1]{t_1} \cap Y) \cup U_1$, and similarly $S_2 := (\outArcs[D_2]{s_2} \cap Y) \cup U_2$ and $T_2 := T \cup (U_2 \cap Y)$, noting that $t_1 = s_2$.
   By construction, for $i=1,2$, $\X_i := \setdef{ P \in \X }[ P \subseteq A_i ]$
   is a \almostDisjointTwoPathDecomposition of $\subgraph[i]$.
   For $i=1,2$, define $a_i := |S_i|$ and $b_i := |T_i|$, and let $c_i$ be such that~\ref{FeasibleConfiguration1Path} and~\ref{FeasibleConfiguration2Path} from Definition~\ref{DefinitionFeasibleConfiguration} are satisfied with respect to $D_i$ and $\X_i$.

   Since $S_i \cup T_i \subseteq S \cup T \cup Y$, Property~\ref{FeasibleConfigurationTriangles} in Definition~\ref{DefinitionFeasibleConfiguration} is satisfied for $S_i$, $T_i$ and $\X_i$ for $i=1,2$.
   Moreover, Properties~\ref{FeasibleConfigurationFirst}, \ref{FeasibleConfiguration1Path} and~\ref{FeasibleConfiguration2Path} are satisfied by construction.
   Thus, $(a_i,b_i,c_i)$ is $D_i$-feasible for $i=1,2$, and it remains to check that both configurations satisfy one of the Properties~\ref{FeasibleConfigurationSeriesEmptyEmpty}--\ref{FeasibleConfigurationSeriesCoveredCoveredCombined}.

   Since $D$ was obtained by a series operation, the distance from $s$ to $t$ is at least $2$, and hence $c \neq 1$.
   Note that $a \in \setdef{a_1, a_1-1}$ and $b \in \setdef{b_2, b_2-1}$.
   Moreover, $a = a_1$ if and only if $U_1 \cap Y = \emptyset$, and $b = b_2$ if and only if $U_2 \cap Y = \emptyset$.
   Furthermore, we have
   \begin{multline}
      b_1 - a_2
      = |T_1| - |S_2|
      = |(\inArcs[D_1]{s_1} \cap Y) \cup U_1| - |(\outArcs[D_2]{t_2} \cap Y) \cup U_2| \\
      = \frac{1}{2}|Y| + |U_1 \setminus Y| - \frac{1}{2}|Y| - |U_2 \setminus Y| 
      = |U_1 \setminus Y| - |U_2 \setminus Y|. \label{FeasibleConfigurationSeriesCenterCount}
   \end{multline}
   If $c = 2$, then $U_1$ and $U_2$ must contain the two arcs of the $s_1$-$t_2$-path, i.e., we have $U_1 \cup U_2 \subseteq Y$.
   Thus, $a = a_1 - 1$, $b = b_2 - 1$ and $b_1 = a_2$.
   Furthermore, $S \cap U_1 = \emptyset$ and $T \cap U_2 = \emptyset$ hold.
   Hence, \ref{FeasibleConfigurationSeriesCoveredCoveredCombined} is satisfied.

   Otherwise, i.e., if $c = 0$, the following three cases may occur for each $i = 1,2$:
   $U_i$ is \emph{empty} ($U_i = \emptyset$), $U_i$ is \emph{free} ($\emptyset \neq U_i \subseteq S \cup T$), or $U_i$ is \emph{covered} ($\emptyset \neq U_i \subseteq Y$).
   This results in a total of 9 combinations that correspond to Properties~\ref{FeasibleConfigurationSeriesEmptyEmpty}--\ref{FeasibleConfigurationSeriesCoveredCoveredUncombined}.
   The mapping is given by the property labels, e.g., (EF) means that $U_1$ is empty and $U_2$ is free, and (FC) means that $U_1$ is free and $U_2$ is covered.
   We briefly discuss the involved configurations.
   The sets $U_i$ determine the values of $a - a_1$ and $b - b_2$ and, by~\eqref{FeasibleConfigurationSeriesCenterCount}, the value of $b_1 - a_2$.
   Finally, $c = 0$ if $U_i$ is empty, and $c = 1$ otherwise.
   
   To see sufficiency, consider a configuration $(a,b,c)$ and $D_i$-feasible configurations $(a_i,b_i,c_i)$ for $i = 1,2$ such that one of the Properties~\ref{FeasibleConfigurationSeriesEmptyEmpty}--\ref{FeasibleConfigurationSeriesCoveredCoveredCombined} is satisfied.
   Hence, there exist sets $S_i \subseteq \outArcs[D_i]{s_i}$, $T_i \subseteq \inArcs[D_i]{t_i}$ and \almostDisjointTwoPathDecompositions $\X_i$ of $\subgraph[i]$ that satify Properties~\ref{FeasibleConfigurationFirst}--\ref{FeasibleConfigurationLast} of Definition~\ref{DefinitionFeasibleConfiguration} with respect to SP-digraph $D_i$.
   We show that $(a,b,c)$ is $D$-feasible.

   For $i = 1,2$, let again $U_i := \setdef{ (s_i,t_i) \in A_i }$ be the set containing the arc $(s_i,t_i)$ if such an arc exists in $A_i$ (and $U_i = \emptyset$ otherwise).
   Note that $U_i \subseteq S \cup T \cup Y$ since the arc cannot be part of a path in $D_i$.
   We will now define two sets of arcs $T^* \subseteq T_1$ and $S^* \subseteq S_2$ that have the same cardinality.
   This allows us to augment $\X_1 \cup \X_2$ by adding $|T^*| = |S^*|$ many paths consisting of exactly one arc from $T^*$ and exactly one arc from $S^*$.
   We distinguish several cases, depending on which property is satisfied by $(a_i,b_i,c_i)$:
   \begin{itemize}
   \item
      Case 1:
      Properties~\ref{FeasibleConfigurationSeriesEmptyEmpty}, \ref{FeasibleConfigurationSeriesFreeEmpty},
      \ref{FeasibleConfigurationSeriesEmptyFree},
      or \ref{FeasibleConfigurationSeriesFreeFree} are satisfied: \\
      Set $T^* := T_1 \setminus U_1$, $S^* := S_2 \setminus U_2$, $S := S_1$ and $T := T_2 $.
   \item
      Case 2: Properties~\ref{FeasibleConfigurationSeriesCoveredEmpty} or~\ref{FeasibleConfigurationSeriesCoveredFree} are satisfied. \\
      Set $T^* := T_1$, $S^* := S_2 \setminus U_2$, $S := S_1 \setminus U_1$ and $T := T_2 $.
   \item
      Case 3: Properties~\ref{FeasibleConfigurationSeriesEmptyCovered} or~\ref{FeasibleConfigurationSeriesFreeCovered} are satisfied. \\
      Set $T^* := T_1 \setminus U_1$, $S^* := S_2$, $S := S_1 $ and $T := T_2\setminus U_2$. 
   \item
      Case 4: Properties~\ref{FeasibleConfigurationSeriesCoveredCoveredUncombined} or~\ref{FeasibleConfigurationSeriesCoveredCoveredCombined} are satisfied. \\
      Set $T^* := T_1$, $S^*$, $S := S_1 \setminus U_1$ and $T := T_2 \setminus U_2$.
   \end{itemize}

   Let $c := 2$ if Property~\ref{FeasibleConfigurationSeriesCoveredCoveredCombined} is satisfied, and $c := 0$ otherwise.
   It is easy to check that in all cases we have $|T^*| = |S^*|$.
   We now consider an arbitrary set $\Y$ of $|T^*|$-many disjoint $2$-paths in $T^* \cup S^*$ having one arc in $T^*$ and the other arc in $S^*$.
   We claim that the union $\X := \X_1 \cup \X_2 \cup \Y$ is a \almostDisjointTwoPathDecomposition of $\subgraph$.
   Since $D$ is simple, no pair of paths in $\Y$ shares more than one vertex, and hence $\Y$ is itself a \almostDisjointTwoPathDecomposition of the corresponding subgraph.
   Moreover, any path in $\X_1$ shares the vertex $t_1 = s_2$ or no vertex at all with any path in $\X_2$.
   Suppose a path $u$-$v$-$w$ in $\X_1$ shares two vertices with a path $p$-$t_1$-$q$ in $\Y$.
   Since $q \notin V_1$, we have $u = p$ and $w = t_1$, which violates Property~\ref{FeasibleConfigurationTriangles} for $\X_1$.
   Similarly, suppose a path $u$-$v$-$w$ in $\X_2$ shares two vertices with a path $p$-$s_2$-$q$ in $\Y$.
   Since $p \notin V_2$, we have $w = q$ and $u = s_2$, which violates Property~\ref{FeasibleConfigurationTriangles} for $\X_2$.
   We obtain that $(a,b,c)$ is a $D$-feasible configuration of $\subgraph$, which concludes the proof.
\end{proof}

We now turn to the parallel composition.

\begin{lemma}[feasible configurations for parallel composition]
   \label{TheoremFeasibleConfigurationsParallel}
   Let $D = (V,A)$ be a simple SP-digraph with source $s$ and sink $t$ such that $D$ is obtained 
   from SP-digraphs $D_i = (V_i,A_i)$ with sources $s_i$ and sinks $t_i$ for $i=1,2$ by a parallel operation. 
   Then a configuration $(a,b,c)$ is $D$-feasible if and only if there exist $D_i$-feasible configurations $(a_i,b_i,c_i)$ for $i=1,2$ that satisfy 
   \begin{enumerate}[label={(\roman*)}]
   \item
      \label{FeasibleConfigurationParallelTriple} $(a,b,c) = (a_1 + a_2, b_1 + b_2, \max \setdef{c_1,c_2})$ and
   \item
      \label{FeasibleConfigurationParallelC} $c_1 = 0$ or $c_2 = 0$.
   \end{enumerate}
\end{lemma}
  
\begin{proof}
   To see necessity, let $(a,b,c)$ be a $D$-feasible configuration.
   By definition, there exist sets $S \subseteq \outArcs[D]{s}$, $T \subseteq \inArcs[D]{t}$ and a \almostDisjointTwoPathDecomposition $\X$ of $\subgraph$ that satisfy Properties~\ref{FeasibleConfigurationFirst}--\ref{FeasibleConfigurationLast} of Definition~\ref{DefinitionFeasibleConfiguration}.
   By construction, for $i=1,2$, $\X_i := \setdef{ P \in \X }[ P \subseteq A_i ]$ is a \almostDisjointTwoPathDecomposition of $\subgraph[i]$ with $S_i = S \cap A_i$ and $T_i = T \cap A_i$.
   For $i=1,2$, define $a_i := |S_i|$, $b_i := |T_i|$ and $c_i \in \setdef{0,1,2}$ accordingly.
   It is now easy to verify~\ref{FeasibleConfigurationParallelTriple} componentwise.
   Since $D$ is simple, at most one subgraph may contain arc $(s,t)$.
   Moreover, at most one of the decompositions $\X_i$ may contain an $s$-$t$-path, and in this case, the other subgraph may not contain an $(s,t)$-arc.
   Together, this shows~\ref{FeasibleConfigurationParallelC}.

   To see sufficiency, consider a configuration $(a,b,c)$ and $D_i$-feasible configurations $(a_i,b_i,c_i)$ for $i = 1,2$ such that~\ref{FeasibleConfigurationParallelTriple} and~\ref{FeasibleConfigurationParallelC} are satisfied.
   Hence, there exist sets $S_i \subseteq \outArcs[D_i]{s_i}$, $T_i \subseteq \inArcs[D_i]{t_i}$ and \almostDisjointTwoPathDecompositions $\X_i$ of $\subgraph[i]$ that satify Properties~\ref{FeasibleConfigurationFirst}--\ref{FeasibleConfigurationLast} of Definition~\ref{DefinitionFeasibleConfiguration} with respect to SP-digraph $D_i$.
   We show that $(a,b,c)$ is $D$-feasible.

   To this end, let $\X := \X_1 \cup \X_2$, $S := S_1 \cup S_2$ and $T := T_1 \cup T_2$.
   Properties~\ref{FeasibleConfigurationCardinalities}, \ref{FeasibleConfiguration1Path} and~\ref{FeasibleConfiguration2Path} of Definition~\ref{DefinitionFeasibleConfiguration} are easily verified.
   Suppose Property~\ref{FeasibleConfigurationTriangles} is violated for an arc $(u,v) \in S \cup T$.
   This means that for some $i \in \setdef{1,2}$ there exists a $u$-$v$-path in $\X_i$.
   Since $(a_i,b_i,c_i)$ is $D_i$-feasible, $(u,v)$ must be part of the other digraph $D_j$ with $j \neq i$, which implies $u = s$ and $v = t$.
   As a consequence, $c_i = 1$ and $c_j = 2$, which contradicts~\ref{FeasibleConfigurationParallelC}.
   This concludes the proof.
\end{proof}

\DeclareDocumentCommand\D{}{\mathcal{D}}
\DeclareDocumentCommand\T{}{\mathcal{T}}
\DeclareDocumentCommand\rootGraph{m}{\rootGraphOperator(#1)}

\begin{definition}[SP-decomposition tree of an SP-digraph, see Chapter~11 in~\cite{brandstadt:graphClasses}]
   For every SP-digraph $D$ there is a rooted binary tree $\T$ (the SP-decomposition tree) whose leafs are associated with SP-digraphs consisting of two nodes that are connected by an arc, and whose non-leaf nodes are labeled either ``series'' or ``parallel''.
   They are associated with SP-digraphs that arise by the corresponding operation from the SP-digraphs associated with the children.
   The root node is associated with $D$ and denoted by $\rootGraph{\T}$.
\end{definition}

\begin{algorithm}[H]
\TitleOfAlgo{Feasible configuration detection for SP-digraphs}
\label{AlgoSPFeasibleConfigurationDetection}
\LinesNumbered
\setcounter{AlgoLine}{0}
\AlgoInput{SP-decomposition tree $\T$ for an SP-digraph}
\AlgoOutput{A map $\pi$ which maps the SP-digraphs $D$ of all nodes of $\T$ to their respective sets of feasible configurations.}
\BlankLine
\uIf{$\T$ has only one node}{
  Let $\pi(\rootGraph{\T}) := \setdef{ (1,1,1) }$.
}
\Else{
Let $L$ be the label of $\T$'s root. \;
Let $\T_1$ and $\T_2$ be the induced subtrees of the two root's children;
if $L = $ ``series'' then the sink of $\rootGraph{\T_1}$ shall be the source of $\rootGraph{\T_2}$. \;
\lFor{$i=1,2$}{Compute $\pi_i$ by calling Algorithm~\ref{AlgoSPFeasibleConfigurationDetection} for $\T_i$.} \label{AlgoSPFeasibleConfigurationDetectionRecurse}
Compute the set $Z$ of triples $(a,b,c)$ for which there are $(a_1,b_1,c_1) \in \pi(\rootGraph{\T_1})$ and $(a_2,b_2,c_2) \in \pi(\rootGraph{\T_2})$
that satisfy one of the properties in Lemma~\ref{TheoremFeasibleConfigurationsSeries} if $L = $ ``series'' or
satisfy both properties in Lemma~\ref{TheoremFeasibleConfigurationsParallel} if $L = $ ``parallel''. \label{AlgoSPFeasibleConfigurationDetectionEnumerate} \;
For all SP-digraphs $D$ associated with nodes in $\T$,
let $\pi(D) := \begin{cases}
                  \pi_1(D) & \text{ if $D$ is associated with some node of $\T_1$,} \\
                  \pi_2(D) & \text{ if $D$ is associated with some node of $\T_2$,} \\
                  Z        & \text{ otherwise.}
               \end{cases}$.\;
}
\Return{$\pi$}
\end{algorithm}

\begin{lemma}
   \label{TheoremAlgoSPFeasibleConfigurationDetection}
   Let $D = (V,A)$ be an SP-digraph and let $\T$ be its SP-decomposition tree.
   Then Algorithm~\ref{AlgoSPFeasibleConfigurationDetection} computes in $\orderO{|A| \cdot |V|^4}$ time, for every SP-digraph $D'$ associated with a node of $\T$, the set $\pi(D')$ of $D'$-feasible configurations.
\end{lemma}

\begin{proof}
   We prove the correctness of the algorithm by induction on the number $k$ of $\T$'s nodes.
   The base case $k = 1$ is easy since here $D$ consists of a single arc, say $A = \setdef{(s,t)}$.
   It is easily verified that the only $D$-feasible configuration is $(1,1,1)$.
   If $k \geq 2$, $L$ is either ``series'' or ``parallel''.
   By induction hypothesis, the maps $\pi_1$ and $\pi_2$ that are computed in Line~\ref{AlgoSPFeasibleConfigurationDetectionRecurse} are correct.
   Moreover, Lemmas~\ref{TheoremFeasibleConfigurationsSeries} and~\ref{TheoremFeasibleConfigurationsParallel} ensure that $\pi(D)$ is computed correctly, completing the inductive step.

   The only time-critical step is in Line~\ref{AlgoSPFeasibleConfigurationDetectionEnumerate}, where the $D$-feasible configurations $(a,b,c)$ are computed from the $D_1$-feasible configurations $(a_1,b_1,c_1)$ and $D_2$-feasible configurations $(a_2,b_2,c_2)$ of the digraphs $D_1$ and $D_2$ associated with the root's children.
   Observe that the $a$- and $b$-values are non-negative integers and bounded from above by the degrees of the sources or sinks.
   Hence, for $i=1,2$, $D_i$ can have at most $\orderO{|V|^2}$ different feasible configurations.
   Thus, the computation can be carried out by first enumerating all $D_1$-feasible configurations and all $D_2$-feasible configurations and adding all combinations that can be produced according to Lemma~\ref{TheoremFeasibleConfigurationsSeries} or Lemma~\ref{TheoremFeasibleConfigurationsParallel}, respectively.
   Thus, the whole step can be carried out in $\orderO{|V|^4}$.

   Since $\T$ has at most $|A|$ nodes the whole algorithm runs in time $\orderO{|A| \cdot |V|^4}$.
\end{proof}

Note that for a series composition, the running time of Line~\ref{AlgoSPFeasibleConfigurationDetectionEnumerate} of
Algorithm~\ref{AlgoSPFeasibleConfigurationDetection} can be improved to $\orderO{|V|^3}$
since combinations in which $b_2$ and $a_1$ differ by more than $1$ can be skipped.
Unfortunately, we see no way of improving the step for the parallel composition.

\bigskip

\begin{algorithm}[H]
\TitleOfAlgo{Recursive \almostDisjointTwoPathDecomposition algorithm for SP-digraphs}
\label{AlgoSPRecursive}
\LinesNumbered
\setcounter{AlgoLine}{0}
\AlgoInput{SP-decomposition tree $\T$ for an SP-digraph $D$, a $D$-feasible configuration $(a,b,c)$ and the set $\pi(D')$ of $D'$-feasible configurations for all SP-digraphs $D'$ associated with $\T$'s nodes}
\AlgoOutput{A \almostDisjointTwoPathDecomposition $\X$ of a subgraph of $D$ corresponding to $(a,b,c)$ according to Definition~\ref{DefinitionFeasibleConfiguration}}
\BlankLine
\uIf{$\T$'s root is labeled ``series''}{
   Let $\T_1$ and $\T_2$ be the induced subtrees of the two root's children with associated SP-digraphs $D_1$ and $D_2$ such that the sink $t_1$ of $D_1$ is equal to the source $s_2$ of $D_2$. \;
   Let $s_1$ be $D_1$'s source and $t_2$ be $D_2$'s sink. \;
   Find $(a_1,b_1,c_1) \in \pi(D_1)$ and $(a_2,b_2,c_2) \in \pi(D_2)$ that satisfy one of the properties in Lemma~\ref{TheoremFeasibleConfigurationsSeries}
   with respect to $(a,b,c)$.
   \label{AlgoSPRecursiveSeriesChildConfigurations} \;
   \For{$i=1,2$}{
      Let $\pi_i$ be the restriction of $\pi$ to the nodes of subtree $\T_i$. \;
      Compute $\X_i$ by calling Algorithm~\ref{AlgoSPRecursive} for $(\T_i, (a_i,b_i,c_i), \pi_i)$. \label{AlgoSPRecursiveSeriesChildDecomposition} \;
      Let $S_i := \setdef{ a \in \outArcs[D_i]{s_i} }[ \text{ $a$ is in no arc of $\X_i$ }]$. \;
      Let $T_i := \setdef{ a \in \inArcs[D_i]{t_i} }[ \text{ $a$ is in no arc of $\X_i$ }]$.
   }
   Let $T^* := \begin{cases}
                  T_1 \setminus S_1 & \text{ if $a_1 = a$}, \\
                  T_1               & \text{ otherwise}
               \end{cases} \quad$
   and let $S^* := \begin{cases}
                  S_2 \setminus T_2 & \text{ if $b_2 = b$}, \\
                  S_2               & \text{ otherwise}
               \end{cases}$. \label{AlgoSPRecursiveSeriesMatchingSets} \;
   Compute a \almostDisjointTwoPathDecomposition $\Y$ of $T^* \cup S^*$. \label{AlgoSPRecursiveSeriesMatch} \;
   \Return{$\X_1 \cup \X_2 \cup \Y$.}
}
\uElseIf{$\T$'s root is labeled ``parallel''}{
   Let $\T_1$ and $\T_2$ be the induced subtrees of the two root's children. \;
   Find $(a_1,b_1,c_1) \in \pi(\rootGraph{\T_1})$ and $(a_2,b_2,c_2) \in \pi(\rootGraph{\T_2})$
   that satisfy both properties in Lemma~\ref{TheoremFeasibleConfigurationsParallel} 
   with respect to $(a,b,c)$. \label{AlgoSPRecursiveParallelChildConfigurations} \;
   \For{$i=1,2$}{
      Let $\pi_i$ be the restriction of $\pi$ to the nodes of subtree $\T_i$. \;
      Compute $\X_i$ by calling Algorithm~\ref{AlgoSPRecursive} for $(\T_i, (a_i,b_i,c_i), \pi_i)$ \label{AlgoSPRecursiveParallelChildDecomposition}.
   }
   \Return{$\X_1 \cup \X_2$.}
}
\lElse{\Return{$\emptyset$.}}
\end{algorithm}

\begin{lemma}
   \label{TheoremAlgoSPRecursive}
   Let $D = (V,A)$ be an SP-digraph, let $\T$ be its SP-decomposition tree, $(a,b,c)$ a $D$-feasible configuration and let $\pi(D')$ be the set of $D'$-feasible configurations for all SP-digraphs $D'$ associated with $\T$'s nodes.
   Then Algorithm~\ref{AlgoSPRecursive} works correctly and runs in $\orderO{|V|^3}$ time.
\end{lemma}

\begin{proof}
   We prove the correctness of the algorithm by induction on the number $k$ of $\T$'s nodes.
   If $k = 1$, then $D$ consists of a single arc, say $A = \setdef{(s,t)}$.
   Hence, $(1,1,1)$ is the only $D$-feasible configuration whose corresponding \almostDisjointTwoPathDecomposition in Definition~\ref{DefinitionFeasibleConfiguration} is $\X = \emptyset$.

   If $k \geq 2$, we distinguish two cases.
   \textbf{Case 1:} $\T$'s root is labeled ``series''. \\
   Since $(a,b,c)$ is $D$-feasible, the $D_i$-feasible configurations $(a_i,b_i,c_i)$ in Step~\ref{AlgoSPRecursiveSeriesChildConfigurations} exist by Lemma~\ref{TheoremFeasibleConfigurationsSeries} for $i=1,2$.
   By induction, for $i=1,2$, the \almostDisjointTwoPathDecomposition $\X_i$ computed in Step~\ref{AlgoSPRecursiveSeriesChildDecomposition} corresponds to $(a_i,b_i,c_i)$.
   The sets $S^*$ and $T^*$ constructed in Step~\ref{AlgoSPRecursiveSeriesMatchingSets} have the same cardinality.
   Moreover, a \almostDisjointTwoPathDecomposition $\Y$ of $S^* \cup T^*$ can be obtained by arbitrarily matching the arcs in $T^*$ to arcs in $S^*$ and considering the paths consisting of the matched pairs.
   Finally, Lemma~\ref{TheoremFeasibleConfigurationsSeries} ensures that $\X_1 \cup \X_2 \cup \Y$ is a \almostDisjointTwoPathDecomposition that corresponds to $(a,b,c)$.

\noindent
   \textbf{Case 2:} $\T$'s root is labeled ``parallel''. \\
   Since $(a,b,c)$ is $D$-feasible, the $D_i$-feasible configuration $(a_i,b_i,c_i)$  Step~\ref{AlgoSPRecursiveParallelChildConfigurations} exist by Lemma~\ref{TheoremFeasibleConfigurationsParallel} for $i=1,2$.
   By induction, for $i=1,2$, the \almostDisjointTwoPathDecomposition $\X_i$ computed in Step~\ref{AlgoSPRecursiveParallelChildDecomposition} corresponds to $(a_i,b_i,c_i)$.
   Finally, Lemma~\ref{TheoremFeasibleConfigurationsParallel} ensures that $\X_1 \cup \X_2$ is a \almostDisjointTwoPathDecomposition that corresponds to $(a,b,c)$.

   It remains to prove the statement on the running time
   Steps~\ref{AlgoSPRecursiveSeriesChildConfigurations} and~\ref{AlgoSPRecursiveParallelChildConfigurations} can be carried out in $\orderO{|V|^2}$ since for fixed $a_1$ and $b_2$ (and given $a$, $b$ and $c$), only a constant number of feasible configurations $(a_1,b_1,c_1)$ and $(a_2,b_2,c_2)$ remain.
   As argued above, the computation of $\Y$ is easy and can be done $\orderO{|V|}$.

   Thus, apart from the recursive calls the algorithm runs in $\orderO{|V|^2}$ time.
   Since $\T$ has at most $|V|$ nodes, the whole algorithm runs in time $\orderO{|V|^3}$.
\end{proof}

\bigskip

\begin{algorithm}[H]
\TitleOfAlgo{\almostDisjointTwoPathDecomposition algorithm for SP-digraphs}
\label{AlgoSPDecomposition}
\LinesNumbered
\setcounter{AlgoLine}{0}
\AlgoInput{An SP-digraph $D = (V,A)$}
\AlgoOutput{%
  A \almostDisjointTwoPathDecomposition $\X$ of $D$ or $\emptyset$ if no such decomposition exists
}
\BlankLine
Compute an SP-decomposition tree $\T$ for $D$. \label{AlgoSPDecompositionTree} \;
Compute $\pi$ by calling Algorithm~\ref{AlgoSPFeasibleConfigurationDetection} for $\T$. \label{AlgoSPDecompositionMap} \;
\uIf{$\pi(D)$ contains is a triple $(0,0,c)$ for some $c$ \label{AlgoSPDecompositionCheck}}{
   Compute $\X$ by calling Algorithm~\ref{AlgoSPRecursive} for $(\T, (0,0,c), \pi)$. \;
   \Return{$\X$}
}
\lElse{
   \Return{$\emptyset$.}
}
\end{algorithm}

\begin{theorem}
   \label{TheoremAlgoSPDecomposition}
   Given an SP-digraph $D = (V,A)$, Algorithm~\ref{AlgoSPDecomposition} finds, in $\orderO{|A| \cdot |V|^4}$ time, a \almostDisjointTwoPathDecomposition of $D$ or correctly states that no such decomposition exists.
\end{theorem}

\begin{proof}
   Let $D = (V,A)$ be an SP-digraph.
   Then the decomposition tree $\T$ for $D$ is computed in Step~\ref{AlgoSPDecompositionTree}, which can be done in $\orderO{|A|}$ time with the recognition algorithm from~\cite{ValdesTL82}.
   In Step~\ref{AlgoSPDecompositionMap}, map $\pi$ is computed in $\orderO{|V|^5}$ time (see Lemma~\ref{TheoremAlgoSPFeasibleConfigurationDetection}).
   From Definition~\ref{DefinitionFeasibleConfiguration} it is easy to see that $D$ has a \almostDisjointTwoPathDecomposition if and only if $(0,0,0)$ or $(0,0,2)$ is a $D$-feasible configuration (note that $(0,0,1)$ is never a $D$-feasible configuration).
   This condition is checked in Step~\ref{AlgoSPDecompositionCheck}.
   In the case a decomposition exists, a decomposition corresponding to such a configuration is computed by calling Algorithm~\ref{AlgoSPRecursive} in time $\orderO{|V|^3}$ (see Lemma~\ref{TheoremAlgoSPRecursive}).
   The running time of the algorithm is inherited from that of Algorithm~\ref{AlgoSPFeasibleConfigurationDetection}.
\end{proof}

\paragraph{Extension to undirected graphs.}
One may be tempted to apply similar dynamic programming techniques for undirected graphs as well.
We failed to obtain a dynamic programming algorithm and would like to point out the problems we encountered.

The high-level reason why the dynamic program for digraphs works is that, when combining two partial decompositions of the two subdigraphs (in a series or parallel composition step), it does not matter which arcs were used so far, but only how many arcs were used.
When combining two undirected graphs having source $s$ and sink $t$ using a parallel composition, one may need to pair two edges in order to produce all potential partial \almostDisjointTwoPathDecompositions of the combined graph.
Unfortunately, these pairings are not independent of each other (as it was the case for the series composition for digraphs):
it may happen that one of the edges is an edge $\setdef{s,t}$, and if we want to pair it with an edge $\setdef{v,t}$ we have to ensure that there is no $s$-$v$-path in the partial \almostDisjointTwoPathDecomposition of the graph containing vertex $v$.
Moreover, if the edges $\setdef{v,t}$ and $\setdef{w,t}$ are paired we have to ensure that we do not pair $\setdef{s,v}$ and $\setdef{s,w}$ at the same time.
In order to take care of these problematic cases, additional data has to be stored in the configurations.
Finally, this additional data has to be computed in a series composition, which creates further complications we were unable to handle.